\documentclass[11pt,onecolumn,letterpaper]{article}

\usepackage{color,soul}
\usepackage{cite}
\usepackage{amsmath,amssymb,amsfonts}
\usepackage{textcomp}
\usepackage[pdftex]{graphicx}
\usepackage{microtype}
\usepackage{cancel}
\usepackage{paralist}
\usepackage{comment}
\usepackage{amsthm}
\usepackage[usenames,dvipsnames]{xcolor}
\usepackage[margin=1in]{geometry}
\usepackage{complexity}
\usepackage{xspace}
\usepackage{bm}
\usepackage[T1]{fontenc}
\usepackage{mathtools}
\usepackage{tabularx}
\usepackage{caption}
\usepackage{subfigure}
\usepackage{times}
\usepackage{afterpage}

\usepackage[caption=false,font=footnotesize,labelfont=sf,textfont=sf]{subfig}
\usepackage{framed}
\usepackage[boxed]{algorithm}
\usepackage{array}
\usepackage{color}
\usepackage[noend]{algpseudocode}
\usepackage{ifthen}

\usepackage{tikz}
\usetikzlibrary{calc}

\date{}
\newboolean{short}
\setboolean{short}{false}

\newcommand{\onlyShort}[1]{\ifthenelse{\boolean{short}}{#1}{}}
\newcommand{\onlyLong}[1]{\ifthenelse{\boolean{short}}{}{#1}}
\newcommand{\shortLong}[2]{\ifthenelse{\boolean{short}}{#2}{#1}}
\newcommand{\longShort}[2]{\ifthenelse{\boolean{short}}{#2}{#1}} 

\newcommand{\para}[1]{\vspace{0.2em}\noindent\textbf{#1.}~}
\newcommand{\case}[1]{\vspace{0.2em}\noindent\textbf{#1:}~}

\newcommand{\guessing}{\mathsf{Guessing}\xspace}
\newcommand{\ID}{\mathsf{id}\xspace}

\newcommand{\correct}{\mathsf{Correct}\xspace}

\newcommand{\rtime}{\mathsf{Time}\xspace}
\newcommand{\bad}{\mathsf{Bad}\xspace}
\newcommand{\sLR}{\mathsf{LR}}
\newcommand{\vol}{\mathsf{Vol}\xspace}
\newcommand{\lo}{lo\xspace}
\newcommand{\hi}{hi\xspace}

\newcommand{\cA}{\mathcal{A}\xspace}

\newcommand{\random}{\mathsf{Random}\xspace}

\newcommand{\ra}{\rightarrow}

\makeatletter
\newtheorem*{rep@theorem}{\rep@title}
\newcommand{\newreptheorem}[2]{%
	\newenvironment{rep#1}[1]{%
		\def\rep@title{#2 \ref{##1}}%
		\begin{rep@theorem}}%
		{\end{rep@theorem}}}
\makeatother

\renewcommand{\le}{\leqslant}
\renewcommand{\ge}{\geqslant}
\renewcommand{\geq}{\geqslant}
\renewcommand{\leq}{\leqslant}

\newcommand{\suchthat}{\bigm|}

\algblockdefx{MRepeat}{EndMRepeat}{\textbf{repeat}}{}
\algnotext{EndMRepeat}

\algblockdefx{Repeat}{EndRepeat}{\textbf{repeat }}{}
\algnotext{EndRepeat}

\algblockdefx{MIf}{EndMif}{\textbf{if }}{}
\algnotext{EndMif}

\algblockdefx{MFor}{EndMFor}{\textbf{for }}{}
\algnotext{EndMFor}

\usepackage{enumitem}
\usepackage{multirow}
\usepackage{url}

\newtheorem{theorem}{Theorem}
\newtheorem{observation}[theorem]{Observation}
\newtheorem{lemma}[theorem]{Lemma}
\newtheorem{corollary}[theorem]{Corollary}
\newtheorem{claim}[theorem]{Claim}
\newtheorem{definition}[theorem]{Definition}

\newreptheorem{lemma}{Lemma}
\newreptheorem{theorem}{Theorem}
\newreptheorem{corollary}{Corollary}

\let\originalleft\left
\let\originalright\right
\renewcommand{\left}{\mathopen{}\mathclose\bgroup\originalleft}
\renewcommand{\right}{\aftergroup\egroup\originalright}

\newcommand{\Prob}[1]{\mathrm{Pr}\left[#1\right]\xspace}

\newcommand{\expect}[1]{\mathbb{E}\left[#1\right]\xspace}

\def\BibTeX{{\rm B\kern-.05em{\sc i\kern-.025em b}\kern-.08em
    T\kern-.1667em\lower.7ex\hbox{E}\kern-.125emX}}

\begin{document}

\title{Slow links, fast links, and the cost of gossip }

\author{
Suman Sourav \\
{\em National University of Singapore}\\
\textsf{sourav@comp.nus.edu.sg}
\and
Peter Robinson\\
{\em Royal Holloway, University of London}\\
\textsf{peter.robinson@rhul.ac.uk}
\and
Seth Gilbert \\
{\em National University of Singapore}\\
\textsf{seth.gilbert@comp.nus.edu.sg}
}
\maketitle

\begin{abstract}
Consider the classical problem of information dissemination: one (or more) nodes in a network have some information that they want to distribute to the remainder of the network.  In this paper, we study the cost of information dissemination in networks where edges have latencies, i.e., sending a message from one node to another takes some amount of time.  We first generalize the idea of conductance to weighted graphs by defining $\phi_*$ to be the ``critical conductance'' and $\ell_*$ to be the ``critical latency''. %
One goal of this paper is to argue that $\phi_*$ %
characterizes the connectivity of a weighted graph with latencies in much the same way that conductance characterizes the connectivity of unweighted graphs. %

We give near tight lower and upper bounds on the problem of information dissemination, up to polylogarithmic factors.  Specifically, we show that in a graph with (weighted) diameter $D$ (with latencies as weights) and maximum degree $\Delta$, any information dissemination algorithm requires at least $\Omega(\min(D+\Delta, \ell_*/\phi_*))$ time %
in the worst case. We show several variants of the lower bound (e.g., for graphs with small diameter, graphs with small max-degree, etc.) by reduction to a simple combinatorial game.  

We then give nearly matching algorithms, showing that information dissemination can be solved in $O(\min((D+\Delta)\log^3{n}, (\ell_*/\phi_*)\log n)$ time. %
This is achieved by combining two cases. %
We show that the classical push-pull algorithm is (near) optimal when the diameter or the maximum degree is large.  For the case where the diameter and the maximum degree are small, we give an alternative strategy in which we first discover the latencies and then use an algorithm for known latencies based on a weighted spanner construction. (Our algorithms are within polylogarithmic factors of being tight both for known and unknown latencies.)

While it is easiest to express our bounds in terms of $\phi_*$ and $\ell_*$, in some cases they do not provide the most convenient definition of conductance in weighted graphs. Therefore we give a second (nearly) equivalent characterization, namely the average conductance $\phi_{avg}$. \hspace{-7mm}

\end{abstract}

\newpage
\section{Introduction}

Consider the problem of disseminating information in a large-scale distributed system: a source node in the network has some information that it wants to share/aggregate/reconcile with others. This fundamental problem has been widely studied under various names, e.g., information dissemination (e.g., \cite{Censor-Hillel:2012:GCP:2213977.2214064}), rumor spreading (e.g., \cite{Chierichetti:2010:RSG:1873601.1873736}), global broadcast (e.g., \cite{Haeupler:2013:SFD:2627817.2627868}), one-to-all multicast, information spreading (e.g., \cite{weak_conductance}), and {gossip} (e.g., \cite{Haeupler:2014:OGD:2611462.2611489}).  

Real world network communication often has a time delay, which we model here as edges with latencies. The \emph{latency} of an edge captures how long communication takes, i.e., how many rounds it takes for two neighbors to exchange information. Low latency on links imply faster message transmission whereas higher latency implies longer delays.

In the case of unweighted graphs, all edges are considered the same and are said to have unit latencies. However, this is not true in real life and link latencies can vary greatly.  
In fact, even if nodes are connected directly it might not be the fastest route for communication due to large latency of the link (which might arise due to poor connection quality, hardware or software restrictions etc.); often choosing a multi-hop lower latency path leads to faster distribution of information.

For unweighted graphs (without latencies), there exists a significant amount of literature, characterizing the connectivity of a graph (referred as the conductance of a graph) which exactly indicates how efficient information dissemination will be. We would like to do the same for graphs with latencies, however, due to the presence of latencies, not all edges can be regarded as the same; and therefore connectivity alone is no longer enough.
The usual notion of (unweighted) conductance no longer characterizes the efficiency (or bottleneck) of communication in a graph with latencies.\footnote{Notice you might model an edge with weight $w$ as a path of $w$ edges with weight $1$. If you calculate the conductance of the resulting graph, you do not get a good characterization of the connectivity of the original graph for a few different reasons. For instance, consider the ability of the imaginary nodes on the edge to pull data from the endpoints.}  Thus, we introduce a new notion of the critical weighted conductance $\phi_*$ that generalizes the notion of classical conductance. 
Using $\phi_*$, we give nearly tight lower and upper bounds for information dissemination. 

For some cases, $\phi_*$ might not be the most convenient definition of conductance in weighted graphs. Alternatively, we give a (nearly) equivalent characterization, namely the average weighted conductance $\phi_{avg}$.

\para{Model} We model the network as a connected, undirected graph $G = (V,E)$ with $n = |V|$ nodes.  Each node knows the identities of its neighbors and a polynomial upper bound on the size of the network. Nodes communicate bidirectionally over the graph edges, and communication proceeds in synchronous rounds. An
edge is said to be activated whenever a node sends any message over the edge.
Latencies occur in the communication channel and not on the nodes. For simplicity, we assume that each edge latency is an integer.  (If not, latencies can be scaled and rounded to the nearest integer.) Also, the edge latencies here are symmetric. Problems for non-symmetric arbitrarily large latencies are at least as hard as directed unweighted networks (for which many tasks are impossible to achieve efficiently). Let $D$ be the (weighted) diameter of the graph (with latencies as weights), and let $\ell_\text{max}$ be the maximum edge latency.  We consider both cases where nodes know the latencies of adjacent edges (Section \ref{sec:known}) and cases where nodes do not know the latencies of adjacent edges (the rest of the paper).  Nodes do not know $D$ or $\ell_\text{max}$.\footnote{In real world settings, nodes are often aware of their neighbors.  However, due to fluctuations in network quality (and hence latency), a node cannot necessarily predict the latency of a connection.} 

In each round, each node can choose one neighbor to exchange information with: it sends a message to that neighbor and (automatically) receives a response.\footnote{Notice that this model of communication is essentially equivalent to the traditional \emph{push-pull} where each node can either push data to a neighbor or pull data from a neighbor; here we assume a node always does both simultaneously.  Without the ability to pull data, it is easy to see that information exchange takes $\Omega(nD)$ time, e.g., in a star. Simple flooding matches this lower bound.} If the edge has latency $\ell$, then this round-trip exchange takes time $\ell$. {This model is, within constant factors, equivalent to a more standard model in which a round-trip involves first sending a message with latency $\ell$, receiving it at the other end, and then sending a response at a cost of latency $\ell$.} 
Notice that each node can initiate a new exchange in every round, even if previous messages have not yet been delivered, i.e., communication is non-blocking.

\para{Information dissemination}  %
Here, we mainly focus on \emph{one-to-all information dissemination}. 
A designated source node begins with a message (the \emph{rumor}) and, when the protocol completes, every node should have received the message.

Classic examples include distributed database replication, sensor network data aggregation, and P2P publish-subscribe systems. This fundamental problem has been widely studied under various names, e.g., information dissemination (e.g., \cite{Censor-Hillel:2012:GCP:2213977.2214064}), rumor spreading (e.g., \cite{Chierichetti:2010:RSG:1873601.1873736}), global broadcast (e.g., \cite{Haeupler:2013:SFD:2627817.2627868}), one-to-all multicast, and information spreading (e.g., \cite{weak_conductance}). 
 As a building block, we look at \emph{local broadcast}, i.e., the problem of every node distributing a message to all of its neighbors. 
 
\para{Conductance in weighted graphs}
Our goal in this paper is to determine how long it takes to disseminate information in a graph with latencies. Clearly the running time will depend on the (weighted) diameter $D$ of the graph. %
Typically, such algorithms also depend on how well connected the graph is, and this is normally captured by the conductance $\phi$.  Unfortunately, conductance is no longer a good indicator of connectivity in a graph with latencies, as slow edges (with large weights) are much worse than fast edges.

We begin by generalizing the idea of conductance to weighted graphs.  We give two (nearly) equivalent definitions of conductance in weighted graphs, which we refer to as the critical weighted conductance $\phi_*$ (Definition \ref{def:weightedcond}) and the average weighted conductance $\phi_{avg}$ (Definition \ref{def:avgweightedcond}).  While they give (approximately) the same value for every graph, there are times when one definition is more convenient than the other.  In fact, we show that the values of $\phi_*$ and $\phi_{avg}$ are closely related; as in $\frac{\phi_*}{2\ell_*} < \phi_{avg} < \frac{\phi_*}{\ell_*}\lceil\log (\ell_{max})\rceil $ (c.f. Theorem \ref{thm:conductance}).  We compare these definitions further in Section \ref{subsec:comparision}.  We use $\phi_*$ in determining the lower and upper bounds for information dissemination as it makes our analysis simpler and then use the above relation to determine the bounds for $\phi_{avg}$.

A core goal of this paper is to argue that the notion of $\phi_*$ (and $\phi_{avg}$) defined herein well captures the connectivity of weighted graphs, and may be useful for understanding the performance of other algorithms.

\para{Lower bounds}
These constitute some of the key technical contributions of this paper.
For a graph $G$, with diameter $D$, maximum degree $\Delta$, critical weighted conductance $\phi_*$, and critical latency $\ell_*$, we show that any information dissemination algorithm requires $\Omega(\min(D+\Delta, {\ell_*}/{\phi_*}))$ rounds. That is, in the worst case it may take time $D+\Delta$ to distribute information. However, if the graph is well connected, then we may do better and the time is characterized by the critical weighted conductance.  We show that this lower bound holds even in various special cases, e.g., for graphs with small diameter, or with small max-degree, etc. By the relation provided in Theorem \ref{thm:conductance}, we determine the lower bound in terms of average weighted conductance as $\Omega(\min(D+\Delta, {1}/{\phi_{avg}}))$. 

The main technique we use for showing our lower bounds is a reduction to a simpler combinatorial guessing game.  (See~\cite{DBLP:conf/wdag/Newport14} for a demonstration of how other variants of guessing games can be used to prove lower bounds for radio networks.)  We first show that the guessing game itself takes a large number of rounds. Thereafter we reduce the problem of solving the game to that of solving information dissemination via a simulation.

\para{Upper bounds}
We then show nearly matching upper bounds, i.e., algorithms for solving information dissemination. In this regard, we differentiate our model into two cases. For the case where nodes are not aware of the adjacent edge latencies, we show that the classical push-pull random phone call algorithm~\cite{Karp} in which each node initiates a connection with a randomly chosen neighbor in each round, completes in $O((\ell_* /\phi_*)\log n)$ rounds. By using the relationship between $\phi_*$ and $\phi_{avg}$, we give a $O((\log(\ell_{max}) /\phi_{avg}) \log n )$ upper bound in terms of $\phi_{avg}$.

For the case where nodes do know the latencies of the incident edges, we obtain nearly tight bounds that are independent of $\Delta$ and $\phi_*$: we give a $O(D \log^3{n})$-time algorithm (which is within polylogarithmic factors of the trivial $\Omega(D)$ lower bound). The key idea of the algorithm is to build a (weighted) spanner (based on that in \cite{baswana}). This spanner is then used to distribute information. {This algorithm, however, requires knowledge of a polynomial upper bound on $n$; hence for completeness we also provide an alternate algorithm in Section} \ref{app:altalgo} that
does not require the knowledge of $n$ but takes an additional $\log D$ factor (instead of $\log n$), making it unsuitable for graphs with large diameters.

Finally, we observe that we can always discover the latencies of the ``important'' adjacent edges in $\tilde O(D+\Delta)$ time\footnote{The notation $\tilde O$ hides polylogarithmic factors, which arise due to $D$ and $\Delta$ being unknown.}, after which we can use the algorithm that works when latencies are known. Hence, even if latencies are unknown, combining the various algorithms, we can always solve the information dissemination in $O(\min((D + \Delta)\log^3{n}, (\ell_*/\phi_*)\log n)$ time (or $O(\min((D + \Delta)\log^3{n}, ( \log(\ell_{max}) /\phi_{avg})\log n)$ time), matching the lower bounds up to polylogarithmic factors (with respect to the critical weighted conductance).

\noindent\textbf{Summary of our contributions.}
To the best of our knowledge, this work provides a first ever characterization of conductance in graphs with latencies. In this regard, we provide two different parameters namely $\phi_*$ and $\phi_{avg}$. Note that, we provide the summary here only in terms of $\phi_*$, however, for each case there exists an alternate version in terms of $\phi_{avg}$.

For lower bounds, we show that there exists graphs with
\begin{compactenum}
\item[(a)] $O(\log n)$ diameter with maximum degree $\Delta$ where local broadcast requires $\Omega(\Delta)$ rounds;
\item[(b)] $O(\ell_*)$ diameter with critical weighted conductance $\phi_*$ where local broadcast requires $\Omega(1/\phi_* +\ell_*)$ rounds;
\item[(c)] $\Theta(1/\phi_*)$ diameter where information dissemination requires $\Omega(\min(D+\Delta, {\ell_*}/{\phi_*}))$ rounds; showing the trade-off among the various parameters affecting information dissemination.
\end{compactenum}
For upper bounds on information dissemination, we show that
\begin{compactenum}
\item[(d)] the push-pull algorithm takes $O(\ell_* \log(n)/\phi_*)$ rounds;
\item[(e)] if nodes are aware of an upper bound on $n$, there exists a spanner-based algorithm that solves information dissemination in $O((D+\Delta)\log^3 n)$ rounds.%
\item[(f)] {if nodes do not know an upper bound on $n$, there exists a pattern-based algorithm that solves information dissemination in $O((D+\Delta)\log^2 n \log D)$.} %
\end{compactenum}

We view our results as a step towards a more accurate characterization of connectivity in networks with delays and we believe that the metrics $\phi_*$ and $\phi_{avg}$ can prove useful in solving other graph problems.

\para{Prior work} 
There is a long history studying the time and message complexity of disseminating information \emph{when all the links have the same latency}.  It is interesting to contrast what can be achieved in the weighted case with what can be achieved in the unweighted case. 

The classic model for studying information dissemination is the \emph{random phone call model}, introduced by~\cite{Demers:1987:EAR:41840.41841}: in each round, each node communicates with a single randomly selected neighbor; if it knows the rumor, then it ``pushes'' the information to its neighbor; if it does not know the rumor, then it ``pulls'' it from its neighbor (see, e.g.,\cite{firege},\cite{Kempe:2001:SGR:380752.380796}, \cite{stacs2011_conductance}). 

An important special case is when the graph is a clique: any pair of nodes can communicate directly. 
In a seminal paper, Karp et al.\cite{Karp} show that a rumor can be disseminated in a complete graph in $O(\log{n})$ rounds with $O(n\log\log{n})$ message complexity. Fraigniaud and Giakkoupis \cite{Fraigniaud:2010:BCC:1810479.1810505} show how to simultaneously achieve optimal communication complexity (except for extremely small rumor sizes). 

When the graph is not a clique, the performance of the classical push-pull protocol, wherein a node exchanges information with a random neighbor in each round, typically depends on the topology of the graph, specifically, how well connected the graph is. An exciting sequence of papers (see~\cite{Mosk-Aoyama:2006:CSF:1146381.1146401, Chierichetti:2010:RSG:1873601.1873736, Chierichetti:2010:ATB:1806689.1806745,stacs2011_conductance}  and references therein) eventually showed that rumor spreading in this manner takes time $O(\frac{\log n}{\phi})$, where $\phi$ is the conductance of the graph.

The question that remained open was whether a more careful choice of neighbors lead to faster information dissemination.  In a breakthrough result, Censor-Hillel et al.~\cite{Censor-Hillel:2012:GCP:2213977.2214064} gave a randomized algorithm for solving information dissemination in any (unweighted) graph in time $O(D + \polylog{n})$, where $D$ here is the non-weighted diameter of the graph. Of note, the protocol has no dependence on the conductance of the graph but only on the diameter (which is unavoidable). 

There were two key ingredients to their solution: first, they gave a ``local broadcast'' protocol where each node exchanges information with \emph{all} its neighbors in $O(\log^3{n})$ time; second,  as a by-product of this protocol they obtain a spanner which they use in conjunction with a simulator (defined therein) to achieve information dissemination in $O(D + \polylog{n})$ time. Haeupler \cite{Haeupler:2013:SFD:2627817.2627868} then showed how local broadcast could be achieved in $O(\log^2{n})$ time using a simple deterministic algorithm.

The conclusion, then, is that in an unweighted graph (with unit latency edges), information dissemination can be achieved in time $O(D + \polylog{n})$ or in time $O(\log ({n})/\phi)$.

\para{Other related works} 
The problem has been well researched in several other settings as well. For graphs modeling social networks Doerr et al. \cite{Doerr:2011:SNS:1993636.1993640, Doerr:2012:WRS:2184319.2184338} show a $\Theta(\log n)$ time bound for solving broadcast. For the case of direct addressing, %
Haeupler and Malkhi \cite{Haeupler:2014:OGD:2611462.2611489} show that  broadcast can be performed optimally in $O(\log \log{n})$ rounds. Information dissemination has been studied in random geometric graphs by Bradonji{\'c} et al. \cite{Bradonjic:2010:EBR:1873601.1873715}, in wireless sensor networks networks by Boyd et al.\cite{Boyd:2006:RGA:1148663.1148679} and Farach-Colton et al. \cite{FARACHCOLTON201360}, in mobile adhoc networks by Fernandez-Anta et al. \cite{FernandezAnta2012} and in dynamic graphs by Sarwate and Dimakis \cite{5062132}, Gandhi et al. \cite{4447471}, and Giakkoupis et al. \cite{Giakkoupis2014}.

\section{Weighted Conductance}

In this section, we consider two different approaches to characterize conductance in weighted graphs, namely, the critical weighted conductance and the average weighted conductance, and we study how these notions relate to each other. In the sections that follow, we focus on the critical weighted conductance for determining the bounds on information dissemination. We obtain corresponding bounds for the average weighted conductance by applying Theorem~\ref{thm:conductance}.

Conductance, in general, is a characterization of the ``bottleneck in communication'' of a graph. {In standard network models, communication or spreading of information can be done faster if the graph is well-connected.} For unweighted graphs, the only bottleneck in communication can be the connectivity of the graph, however, for weighted graphs the bottleneck can arise either due to the graph connectivity or due to the edge latency (even if the nodes are directly connected by a slow edge, there might exist a different multi-hop faster path). Our aim is to capture both aspects of this bottleneck in communication. 

Having good connectivity facilitates faster communication whereas large latencies result in slow-downs.{ Even if a graph is quite well-connected, if most of its edges are slow edges, communication will be slow.}  Ideally, we would want the best connectivity along with the least slowdown for faster communication. We obtain the definition of $\phi_*$ by directly optimizing these orthogonal parameters. 
The connectivity that maximizes this ratio is defined as the critical weighted conductance $\phi_*$ and the corresponding latency is defined as the critical latency $\ell_*$. {In other words, $\phi_*$ captures the critical threshold for which the graph has the best possible connectivity with the least possible slowdown.} %

The definition of the average weighted conductance $\phi_{avg}$ is inspired by the classical notion of conductance. Each cut edge's contribution towards the overall connectivity is normalized by dividing it with its latency (rounded to the upper bound of its latency class), so as to account for the slow-down. 
{Instead of considering the critical threshold, the slowdown here is characterized by the weighted average over the individual slow-down caused by each cut edge.}

\subsection{Critical Weighted Conductance}
\label{sec:unknown_latencies}

We now define the critical weighted conductance of a graph, generalizing the classical notion of conductance.  For a given graph $G = (V,E)$, and for a set of edges $S \subseteq E$, we define $E_\ell(S)$ to be the subset of edges of $S$ that have latency $\le \ell$.  For a set of nodes $U \subseteq V$ and cut $C=(U, V\setminus U)$, we define $E_\ell(C)$ to be the subset of edges across the cut $C$ with latency $\le \ell$, and we define the volume $\vol(U) = \sum_{v\in U}deg_v$, where $deg_v$ is the degree of the node $v$. 

We first define the critical weighted conductance of a cut for a given latency $\ell$, and then define the weight-$\ell$ conductance as the minimum critical weighted conductance across all cuts. 
\begin{definition}[Weight-$\ell$ Conductance]  \label{def:dweightedcond}
	Consider a graph $G = (V,E)$. For any cut $C$ in the set of all possible cuts ($\tilde{C}$) of the graph $G$ and an integer $\ell$, we define 
	\begin{align*}
		\phi_\ell(C) = \frac{|E_\ell(C)|}{\min\{\vol(U),\vol(V\setminus U)\}}\ .
	\end{align*}
	The \emph{weight-$\ell$ conductance} is given by $\phi_\ell(G) = \min \{ \phi_{\ell}(C) \mid {C \in \tilde{C}} \}$.
\end{definition}

\begin{definition}[Critical Weighted Conductance]
	\label{def:weightedcond}
	We define the \emph{critical weighted conductance $\phi_*(G)$} as 
	
\[ \phi_*(G) =  \phi_\ell(G) \mathrel{}\biggr\rvert\mathrel{} \frac{\phi_\ell(G)}{\ell} \text{ is maximum for any } \ell \in (1,\ell_{max}). 
\]	

\noindent We call $\ell_*$ the \emph{critical latency} for $G$ if $\ell_* = \ell$ and $\phi_*(G) = \phi_\ell(G)$. \end{definition}
We simply write $\phi_*$ (or $\phi_\ell$) instead of $\phi_*(G)$ (or $\phi_\ell(G)$) when graph $G$ is clear from the context.
If all edges have latency $1$, then $\phi_*$ is exactly equal to the classical graph conductance~\cite{Jerrum:1988:CRM:62212.62234}. 

\subsection{Average Weighted Conductance} 

For a given graph $G = (V,E)$, we first define $\lceil\log (\ell_{max})\rceil$ different latency classes, where the first class contains all the edges of latency $\leq 2$ and the subsequent $i^{th}$ latency class consists of all the edges in the latency range of $(2^{i-1},2^i]$. 
For a set of nodes $U \subseteq V$ and the cut $C=(U, V\setminus U)$, we define $k_i(C)$ to be the subset of edges across the cut $C$ belonging to latency class $i$ (i.e. all cut edges of latency $>2^{i-1}$ and $\leq 2^i$).

For a cut $C$, we first define the average cut conductance  as $\phi_{avg}(C)$, and then define the average weighted conductance as the minimum average cut conductance across all cuts. 
\begin{definition}[Average Cut Conductance]  \label{def:avgcutcond}
  Consider a graph $G = (V,E)$, a set of nodes $U \subseteq V$ and the cut $C = (U, V\setminus U)$. Let $S$ be the $\min\{\vol(U),\vol(V\setminus U)\}$. %
\begin{align*}
  \phi_{avg}(C) = \frac{1}{S} \sum\limits_{i=1}^{\lceil\log (\ell_{max})\rceil} \frac{|k_i(C)|}{2^i}
\end{align*}

\end{definition}

\begin{definition}[Average Weighted Conductance]  \label{def:avgweightedcond}
Let $\tilde{C}$ be the set of all possible cuts of the graph $G$.
We define the \emph{average weighted conductance} as  $\phi_{avg}(G) = \min \{ \phi_{avg}(C) \mid {C \in \tilde{C}} \}$.
\end{definition}
We simply write $\phi_{avg}$ instead of $\phi_{avg}(G)$ when graph $G$ is clear from the context.
If all edges have latency $1$, then $\phi_{avg}$ is exactly half of the value of the classical graph conductance.%

\subsection{Comparing $\phi_*$ and $\phi_{avg}$} \label{subsec:comparision}

Surprisingly, we see that $\phi_*$ and $\phi_{avg}$ are closely related and to show the relationship between them, we first define $\mathcal{L}$ as the number of non-empty latency classes in the given graph $G$. 
Latency class $i$ is said to be non-empty if there is at least one edge in the graph $G$ that has a latency $>2^{i-1}$ and $\leq 2^i$. The maximum value that $\mathcal{L}$ can take is $\lceil\log (\ell_{max})\rceil$ which is the total number of possible latency classes.

\begin{theorem} \label{thm:conductance}
\[ \frac{\phi_*}{2\ell_*} \leq \phi_{avg} \leq \mathcal{L}\frac{\phi_*}{\ell_*} \leq \lceil\log (\ell_{max})\rceil \frac{\phi_*}{\ell_*}. \]
\end{theorem}

\begin{proof}
Consider any weighted graph $G$ that has critical weighted conductance $\phi_*$ and critical latency as $\ell_*$. We first show the upper bound. Let $C$ be the cut from which $\phi_*$ was obtained and let $S$ be the minimum volume among either side of the cut. 
By the definition of weight-$\ell$ conductance, $\phi_{2^i}(C) = ({\sum_{j=1}^{i}|k_j(C)|})/{S} $
\[ \Rightarrow \frac{\phi_{2^i}(C)}{2^i} = \frac{\sum_{j=1}^{i}|k_j(C)|/2^i}{S} \geq \frac{|k_i(C)|/2^i}{S} \]

and from the definition of $\phi_*$, we know that $\frac{\phi_*}{\ell_*}$ is $\geq \frac{\phi_\ell}{\ell}$  for any $\ell$, which implies
\[ \forall i \in (1,\lceil\log (\ell_{max})\rceil) \hspace{.4cm} \frac{\phi_*}{\ell_*} \geq \frac{\phi_{2^i}(C)}{2^i} \geq \frac{|k_i(C)|/2^i}{S}    \]

Note that, in the definition of $\phi_{avg}$, the terms corresponding to the empty latency classes becomes zero.
We replace each remaining term in the definition of $\phi_{avg}(C)$ by $\frac{\phi_*}{\ell_*}$ and using the above inequality, we get $\phi_{avg}(C) \leq \frac{\phi_*}{\ell_*} \mathcal{L}$. Combining with the fact that $\phi_{avg}$ is the minimum average cut conductance, we obtain %
\begin{equation}\label{eq:cutC}
\phi_{avg} \leq \phi_{avg}(C) \leq \mathcal{L} \frac{\phi_*}{\ell_*}  \leq \lceil\log (\ell_{max})\rceil \frac{\phi_*}{\ell_*}
\end{equation}

Next we show the lower bound, for this we consider the cut $C'$ that determines $\phi_{avg}$ and let $S'$ be the minimum volume among either side of the cut. On this cut $C'$ consider the latency class of the critical latency $\ell_*$; say $\ell_*$ lies in the latency class $x$, which implies that $2^{x-1} < \ell_* \leq 2^x$. From the definition of weight-$\ell$ conductance, we get 
\[\frac{\phi_{\ell_*} (C')}{2\ell_*} \leq \frac{|k_1(C')|+|k_2(C')|+ \dots + |k_x(C')|}{2\ell_* S'} .\]
Rewriting $\phi_{avg}$ as (from definition)
\[ \phi_{avg} = \frac{|k_1(C')|}{2S'}+\frac{|k_2(C')|}{2^2S'}+ \dots + \frac{|k_{\lceil\log (\ell_{max})\rceil}(C')|}{2^{\lceil\log (\ell_{max})\rceil} S'}, \]
and comparing the first $x$ terms of $\phi_{avg}$ to that of  ${\phi_{\ell_*} (C')}/{2\ell_*}$, we observe that each term in the expression of $\phi_{avg}$ is at least as large as the corresponding term in the above upper bound on  ${\phi_{\ell_*} (C')}/{2\ell_*}$. Also there are some additional positive terms in $\phi_{avg}$. Combining this with the fact that ${\phi_{\ell_*}}/{2\ell_*} \leq {\phi_{\ell_*} (C')}/{2\ell_*}$ (as by definition $\phi_{\ell}$ is chosen as the minimum value among all possible cuts), we obtain  %
\begin{equation}\label{eq:cutC'}
\frac{\phi_{\ell_*}}{2\ell_*} \leq \frac{\phi_{\ell_*} (C')}{2\ell_*} \leq \phi_{avg} 
\end{equation}

This proves the lower bound and completes the proof.
\end{proof}

\section{Lower Bounds}

We proceed to lower bound the time for completing information dissemination.  The main goal of this section (as found in Theorems~\ref{thm:generalLBCAST}, \ref{thm:generalConductance}, and \ref{thm:PUSHPULLOptimal}) is to show that every gossip algorithm requires $\Omega\left(\min\left\{ \Delta + \D, {\ell_*}/{\phi_*} \right\}\right)$ on graphs with diameter $D$, max-degree $\Delta$, critical weighted conductance $\phi_*$, and critical latency $\ell_*$. %
Throughout this section, we assume that nodes do not know the latencies of their adjacent links ({when nodes do know the latencies, the trivial lower bound of $\Omega(D)$ is sufficient}).

We begin by defining a combinatorial guessing game (a similar approach as in~\cite{DBLP:conf/wdag/Newport14}) and show a lower bound for it.\footnote{The results of \cite{DBLP:conf/wdag/Newport14} do not apply directly to our setting, as their ``proposal set'' of the player must intersect the target set in exactly $1$ element. By contrast, the guessing game here requires us to discover sufficiently many target elements such that every element in the target set occurs at least once.}
We then construct several different worst-case graphs and reduce the guessing game to solving information dissemination on these graphs, thereby showing our lower bound.

\subsection{The Guessing Game} \label{sec:guessing}

We define a guessing game played by Alice against an oracle.  Conceptually, the game is played on a bipartite graph of $2m$ nodes. The oracle selects a subset of the edges as the target.  In each round, Alice guesses a set of at most $2m$ edges, and the oracle reveals any target edges that have been hit. At the same time, if any edge $(u,v)$ in the target set is guessed by Alice, then all adjacent edges $(x,v)$ in the target set are removed from the target set.

\noindent\textbf{Formal Definition:}
Fix an integer $m$.  Let $A$ and $B$ be two disjoint sets of $m$ integers each, i.e., the left and right group of nodes in the bipartite graph.  The winning condition of the game depends on a predicate $P$, which returns a subset of edges from $A \times B$.  For example, $P=\random_p$ returns a subset $T$ that contains elements of $A \times B$, where each element is chosen with probability $p$ or discarded with probability $1-p$.  

We now define the game $\guessing(2m,P)$, which begins when Alice receives two disjoint sets $A$ and $B$.
The oracle chooses a \emph{target set} $T_1 \subseteq A\times B$ returned by the predicate. 
{We use then notation $T^A$ to refer to the projection of $T$ onto set $A$ and define $T^B$ similarly.}
Throughout, we assume that Alice has access to a source of unbiased random bits. 
Alice's goal is to eliminate all the elements in the target set.  In each round $r\ge 1$, Alice submits a set $X_r \subseteq A\times B$ of size at most $2m$ as her \emph{round $r$ guesses} to the oracle.
The oracle replies by revealing the items she guessed correctly, i.e.,  $X_r \cap T_r$.  The oracle then computes the \emph{round $r+1$ target set} $T_{r+1}$ by removing the items that Alice hit, i.e., all the items in $T_r$ that have the same $B$-component as an item in $X_r \cap T_r$: 
\begin{align} \label{eq:update}
	T_{r+1} = T_r \setminus \left(T_r^A\times \left(T_r^B \cap X_r^B \right)\right). %
\end{align} 
This concludes round $r$ and the next round begins.

\noindent\textbf{Winning Condition:}
The game is solved in the first round $r'$, where Alice's guesses result in an empty target set; at this point, the oracle answers \texttt{halt}. 
In other words, the game ends in round $r'$ if, for every $b \in T_1^B$, there was some $a' \in A$ such that $(a',b) \in X_r \cap T_r$, in some $r \in [1,r']$.
Alice's aim is to minimize the number of rounds until the target set becomes empty.
We say that a \emph{protocol $\Pi$ solves $\guessing(2m,P)$ with probability $1 - \epsilon$ in $r$ rounds}, if $\Pi$ always terminates within $r$ rounds, and $T_{r+1}=\emptyset$ with probability $\ge 1 - \epsilon$, for any target set $T$. %
In this case, we call $\Pi$ an \emph{$\epsilon$-error protocol}.

\subsection{Guessing by Simulating Gossiping} \label{sec:gadget}
{We now describe how Alice can devise a guessing game protocol by simulating a distributed gossip algorithm.}
Our gossip lower bound results (see Section~\ref{sec:lbs}), use variants of an $n$-node distributed network that has a \emph{guessing game gadget} of $2m$ nodes embedded as a subgraph. In our gadget construction, we use predicate $P$, to specify a set of hidden low latency edges, which we call \emph{fast edges}. We show that the execution of a gossip algorithm on an $n$-node network can be simulated by Alice when playing the guessing game $\guessing(2m,P)$, where $n\ge 2m$. 
More specifically, this construction ensures that Alice can solve the guessing game in $T$ rounds if the distributed algorithm achieves broadcast in $T$ rounds (see Lemma~\ref{lem:simulation}).

We use the notation $\ID(v)$ to denote the ID of a vertex $v$, which, by construction is unique.
For a given instance of the guessing game, Alice creates a set of nodes $L=\{ v_1,\dots,v_{m} \}$ where $\ID(v_i) = a_i \in A$ for $i=1,\dots,m$ and, similarly, maps the integers in $B$ to the IDs of the vertex set $R=\{u_1,\dots,u_m\}$ in a one-to-one fashion. Next, Alice creates a complete bipartite graph on sets $L$ and $R$ by adding $m^2$ \emph{cross edges} and adds a clique on the vertices in $L$ where all clique edges are considered to have latency~$1$. 

For given integer parameters $\lo$ and $\hi$, we construct the network in a way that only some cross edges in the target set are useful to the algorithm by giving them a low latency $\lo$ whereas all other cross edges are assigned a large latency value $\hi$.   
Formally, the latencies of a cross edge $e=(v_i,u_j)$ is $\lo$ iff $(\ID(v_i),\ID(u_j)) \in P$; otherwise $e$ has latency $\hi$.
We denote this constructed gadget as $G(2m,\lo,\hi,P)$, where the parameters refer to the size of the gadget (i.e. $2m)$, the low latency value $\lo$, the high latency value $\hi$, and the predicate $P$ respectively. We also use a symmetric variant of the gadget for embedding multiple copies in the network (see Theorem~\ref{thm:PUSHPULLOptimal}), called $G^\text{sym}(2m,\lo,\hi,P)$, where Alice creates a clique on $R$ in addition to the one on $L$. See Fig. \ref{fig:guessing}. %

\begin{figure}
	\centering
	\includegraphics[width=\textwidth/2]{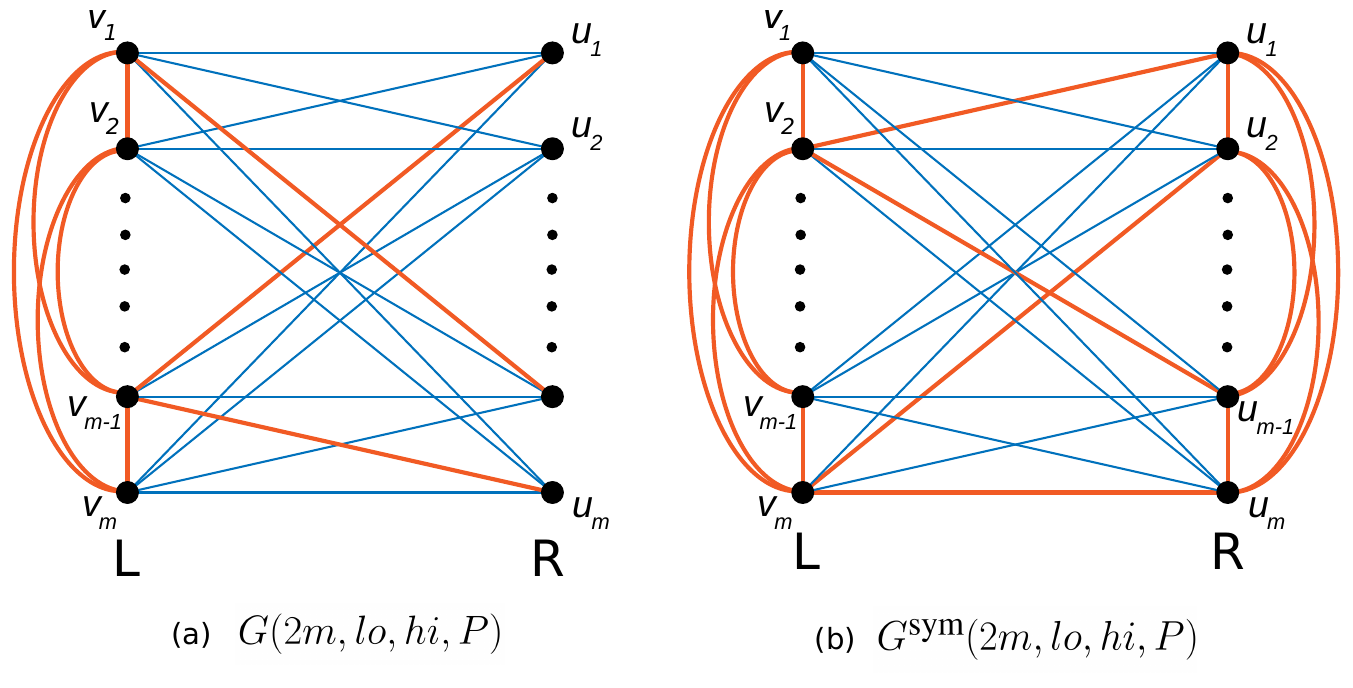}
	\caption{\footnotesize{Guessing Game Gadgets. Red edges correspond to ``fast'' links whereas the blue edges are ``slow'' links with high latency.}}
\label{fig:guessing}
\end{figure}

Since Alice does not know the target set $T$ in advance, she also does not know when a cross edge should have latency~$\lo$ or latency~$\hi$.
Nevertheless, implicitly these latency assignments are fixed a priori by the target set (unknown to Alice) which in turn depends on the predicate $P$.
Whenever a cross edge $e$ is activated in our simulation, Alice submits the ID pair of the vertices of $e$ as a guess to the oracle, whose answer reveals the target set membership and hence also the latency of $e$.

\newcommand{\simLemmaStmt}{Suppose that there is a $t$-round $\epsilon$-error algorithm $\cA$ that solves local broadcast on a given $n$-node network $H$ that contains $G(2m,1,h,P)$ or $G^\text{sym}(2m,1,h,P)$ such that the cross edges of the gadget form a cut of $H$, for $h \ge t$, $n\ge 2m$, and a predicate $P$.
	Then there is an $\epsilon$-error protocol $\Pi$ for $\guessing(2m,P)$ that terminates in $\le t$ rounds.}
\begin{lemma}[Gossip Protocol Simulation] \label{lem:simulation}
	\simLemmaStmt
\end{lemma}

\begin{proof}
We argue that Alice can simulate the execution of $\mathcal{A}$ on network $H$ and, in particular, on the subgraph $G(2m,1,h,P)$, until the gossip algorithm $\cA$ terminates or the oracle answers $\texttt{halt}$. (It is straightforward to extend the argument to  a subgraph $G^\text{sym}(2m,1,h,P)$.)
At the same time, Alice can use the behavior of $\cA$ on the subgraph $G(2m,1,h,P)$ to derive a protocol for $\guessing(2m,P)$.

For a given instance of the guessing game, Alice creates the network $H$ by first assigning all edges in the subgraph $H\setminus G(2m,1,h,P)$ a latency of $1$. Moreover, she creates the edges of the subgraph $G(2m,1,h,P)$ as described in Section \ref{sec:gadget}; we will see below that the latency of a cross edges is only set when it is first activated.

If a non-cross edge $(v_i,v_{j})$ (i.e. a clique edge on $L$ or an edge in $E(H\setminus G(2m,1,h,P))$) is activated by the algorithm, Alice locally simulates the bidirectional message exchange by updating the state of nodes $v_i$ and $v_{j}$ accordingly. 
In each round $r$ of the gossip algorithm, a set of at most $2m$ cross edges is activated by the vertices simulated by Alice.
For each activated cross edge $(v_i,u_j)$, Alice uses $(\ID(v_i),\ID(u_j))$ as one of her round $r$ guesses.
Consider some round $r\ge 1$, and suppose the oracle returns the empty set.
For each one of Alice's submitted round $r$ guess $(a_i,b_j)$ that was not contained in the oracle's answer, Alice sets the latency of $(a_i,b_j)$ to $h$ by updating the local state of $a_i$. Here $a_i = \ID(v_i)$ and $b_j = \ID(u_j)$ that are chosen in round $r$, for some $v_i \in L$ and $u_j \in R$. 
It follows by a simple inductive argument that the state of every vertex in the simulation is equivalent to executing the algorithm on the network.

We now argue that the above simulation of a $t$-round gossip algorithm for local broadcast solves the game $\guessing(2m,P)$ in at most $t$ rounds with probability $\ge 1 - \epsilon$, for any predicate $P$. 
Recall that the guessing game ends if $T$ becomes empty, which happens when Alice's correct guesses have included every $b \in T^B$ at least once.  
By the premise of the lemma, the cross edges of $G(2m,1,h,P)$ form a cut of $H$, which tells us that $\cA$ cannot solve local broadcast without using the cross edges between $L\times R$.
Since every such $b \in R$ is a neighbor of a node in $L$, the only way it can receive a local broadcast message is via a fast cross-edge in $T$.  Hence, if the local broadcast algorithm terminates, we know that $b$ was hit by one of Alice's guesses.
\end{proof}

\subsection{Guessing Game Lower Bounds}

The following lemma is instrumental for showing the $\Omega(\Delta)$ lower bound of Theorem \ref{thm:generalLBCAST}, which holds when there are no other assumptions on the critical weighted conductance of the graph. %

\newcommand{\lemSingleton}{
	Let $\guessing(2m,P(|T|=1))$ be the guessing game where the target set is a single pair chosen uniformly at random from $A \times B$.
	If protocol $\Pi$ is an $\epsilon$-error protocol for $\guessing(2m,P(|T|=1))$ where $\epsilon< 1$, then the number of rounds until $\Pi$ terminates is at least $\Omega(m)$.
}
\begin{lemma} \label{lem:singleton}
	\lemSingleton
\end{lemma}

\begin{proof}
  For the sake of a contradiction, suppose that $\Pi$ solves $\guessing(2m,P(|T|=1))$ in $t < \tfrac{m}{2} -1$ rounds.
We define $\rtime$ to be the random variable giving the number of rounds until termination of $\Pi$. 
Note that $\Prob{ \rtime > t} = 0$ by assumption.

Consider a round $r \le t$ of the protocol and suppose that the game has not yet ended, i.e., Alice has not yet guessed all of $T$ correctly and has made at most $2m(r-1)$ (incorrect) guesses in the previous rounds.
Let $X_r$ denote the (at most $2m$) pairs from $A \times B$ chosen by Alice in round $r$.
Since from Alice's point of view, the adversary has chosen the single element of $T$ uniformly at random from the $m^2$ elements in $A \times B$, the probability that Alice guesses the element of $T$ in round $r$ is at most $\frac{2m}{m^2 - 2m (r - 1)} \le \frac{2}{m - 2r}$.
Let $\correct$ denote the event that the last round of protocol $\Pi$ resulted in an empty target set, i.e, $\Pi$ correctly solves the game. It follows %
\begin{align} \label{eq:guess0}
\scalebox{0.93}[0.93]{$  \Prob{\rtime = r \suchthat\correct } = \Prob{T_r \subseteq X_r   \suchthat \correct} \le \tfrac{2}{m - 2r}.$}
\end{align}

In the remainder of the proof, we will lower bound the probability of event $\{\rtime > t\}$.
If $\rtime>t$, then none of Alice's guesses in rounds $1,\dots,t$ were successful, i.e.,
\begin{align}  
\scalebox{0.9}[0.9]{$\Prob{ \rtime > t \suchthat \correct } \ge  \prod_{i=1}^t\left(1-\Prob{\rtime = i \suchthat\correct}\right)$}. \label{eq:probbnd}
\end{align}
Observe that
\begin{align}
\Prob{\rtime > t } &\ge \Prob{ \rtime > t \suchthat \correct }\ \Prob{ \correct }.
\notag\\
& \ge  \Prob{ \rtime > t \suchthat \correct }\ (1 - \epsilon) \label{eq:probbnd2}
\end{align}
Applying \eqref{eq:guess0} to each round $i\le t$ in \eqref{eq:probbnd} and plugging this into \eqref{eq:probbnd2}, we get 
\begin{align}
	\Prob{\rtime > t } 
		&\ge \prod_{i=1}^t \left( 1 - \frac{2}{m - 2i} \right)\left(1 - \epsilon\right) \notag\\
		& \ge \left( 1 - \frac{1}{\tfrac{m}{2}-t} \right)^{t}\left(1 - \epsilon\right).  \notag
\end{align}

Since the running time of $\Pi$ was assumed to never exceed $t$ rounds, i.e., $\Prob{ \rtime > t} = 0$ and $\epsilon<1$, we get a contradiction to $t < \tfrac{m}{2} - 1$.
\end{proof}
The next lemma bounds the number of guesses required when the target set is less restricted and its edges form a random subset of the cross edges between $A\times B$. This allows us to derive a lower bound on the local broadcast time complexity in terms of the critical weighted conductance in Theorem \ref{thm:generalConductance}.

\begin{lemma}\label{lem:randomGuessing}
	For the guessing game input sets $A$ and $B$, let $\random_p$ be the predicate that defines the target set $T$ by adding each element of $A\times B$ to $T$ with probability $p$, for some $p \ge \Omega(\tfrac{1}{m})$.
	Then, the following hold:
	\begin{enumerate}
		\item[(a)] \emph{Any} protocol that solves  $\guessing(2m,\random_p)$ requires 
	$\Omega\left({1}/{p}\right)$ rounds in expectation.
		\item[(b)] If Alice uses the (suboptimal) protocol where she submits her $2m$ guesses in each round by choosing, for each $a \in A$, an element $b' \in B$ uniformly at random, and, for each $b \in B$, an $a' \in A$ uniformly at random, then $\Omega\left(\tfrac{\log m}{p}\right)$ rounds are required in expectation.
	\end{enumerate}
\end{lemma}
Our motivation for considering the suboptimal protocol in the second part of Lemma~\ref{lem:randomGuessing} is its close relation to the push-pull gossip protocol, which we formalize in Theorem~\ref{thm:generalConductance}.
\begin{proof}[Proof of Lemma~\ref{lem:randomGuessing}]
Recall that the game ends when the guesses of Alice have hit each element in $T^B \subseteq B$ at least once, whereas $T^B$ is itself a random variable. 
For the sake of our analysis, we will consider Alice's guesses as occurring sequentially and hence we can assume that elements of $T^B$ are discovered one by one.
For each $j\ge 1$, we define $Z_j$ to denote the number of guesses required to guess the 
$j$-th element of $T^B$, after having already guessed $j-1$ elements.

We will first consider general protocols. 
Considering that each edge is in the target set with probability $p$, we can assume that the target membership of an edge $e$ is determined only at the point when Alice submits $e$ as a guess. Recalling that Alice has full knowledge of the remaining elements in $T^B$ that she still needs to guess (cf.\ \eqref{eq:update}), we can assume that her guess is successful with probability $p$ (as she will only guess edges that potentially discover a new element in $T^B$). For this guessing strategy, this remains true independently of the current target set and the set of previously discovered elements (which we denote by $D_j$). Formally, 
$\Prob{ Z_j \mid D_j,T } = \Prob{ Z_j }$ and hence $\expect{Z_j \mid D_j,T } = \expect{ Z_j } = 1/p $.
Note that any $b \in B$ will be part of some target edge in $T$, i.e., $b\in T^B$, with probability $\ge 1 - (1-p)^m=\Omega(1)$, since $p=\Omega(1/m)$, and therefore
$\expect{|T|} =\Omega(m)$. 
{Let $Y$ be the maximum number of guesses required by Alice's protocol $\Pi$, i.e., $Y = \sum_{i=1}^{|T|} Z_i$.}
It follows that 
 \[
   \expect{ Y } 
      = \expect{\expect{ Y \mid D_j, T }} 
       = \expect{\textstyle{\sum_{i=1}^{|T|}} 
               \expect{ Z_i \mid D_j,T }} \]
\[       = \expect{\textstyle{\sum_{i=1}^{|T|}} 
               \expect{ Z_i }}
       = \Omega(\tfrac{m}{p}).
 \]

Considering that Alice can guess up to $2m$ elements per round, it follows that the time is $\Omega(\tfrac{1}{p})$, which completes the proof for general algorithms.

Now consider the case where Alice uses the protocol where she submits her $2m$ guesses in each round by choosing, for each $a \in A$, an element $b' \in B$ uniformly at random, and, for each $b \in B$, an $a' \in A$ uniformly at random.
Note that this process of selecting her guesses is done obliviously of her (correct and incorrect) guesses so far.

Observe that $Z_j$ depends on a random variable $F_j$, which is the size of $T$ after the $(j-1)$-th successful guess.  
Since $Z_j$ is the number of times that the protocol needs to guess until a \emph{new} element in $T^B$ is discovered, the distribution of $Z_j$ corresponds to a geometric distribution. 
According to Alice's protocol, the probability of guessing a new element is given by $\tfrac{F_j}{m^2}$ and hence 
$  \expect{ Z_j \mid F_j } \geq \tfrac{m^2}{F_j}.$
Let $U = |T_1^B|$; i.e., $U$ is the number of all elements in $B$ that are part of an edge in $T$ initially.
We have
\begin{align*}
  \expect{ Y \mid U\ge \tfrac{m}{2} } 
    &= \expect{\expect{ Y \mid F_i}\mid U\ge \tfrac{m}{2}}  \\
    &= \expect{\textstyle{\sum_{i=1}^{U}} \expect{ Z_i \mid F_i }\mid U\ge \tfrac{m}{2}} \\
    &\ge \sum_{i=1}^{\lfloor m/2 \rfloor} \expect{\expect{ Z_i \mid F_i }\mid U\ge \tfrac{m}{2}} \\
    &\ge \sum_{i=1}^{\lfloor m/2 \rfloor} \frac{m^2}{\expect{F_i \mid U\ge \tfrac{m}{2}}},
\end{align*}
where the last inequality follows from $\expect{1/X} \ge 1/\expect{X}$, for any positive random variable $X$, due to Jensen's Inequality.
Since Alice has already correctly guessed $i-1$ elements from $T^B$, we discard all elements that ``intersect'' with successful guesses when updating the target set at the end of each round, according to \eqref{eq:update}. 
It can happen that the protocol discovers multiple elements of $T^B$ using the round $r$ guesses (which we have assumed to happen sequentially in this analysis).
In that case, the target set is not updated in-between guesses. However, it is easy to see that this does not increase the probability of guessing a new element of $T^B$.
We get
\[
  \expect{F_i \mid U\ge \tfrac{m}{2}} \le (m - i)mp,
\]
and thus 
\[
  \expect{ Y \mid U\ge \tfrac{m}{2}} \ge \frac{m}{p} \sum_{i=1}^{\lfloor m/2 \rfloor} \frac{1}{m -i}.
\]
This sum is the harmonic number $H_{\lfloor m/2 \rfloor -1}$, which is $\Theta(\log m)$, for sufficiently large $m$, and hence
\[
  \expect{Y \mid U\ge \tfrac{m}{2}} \ge \Omega\left(\frac{m\log m}{p}\right).
\]
By the law of total expectation it follows that 
\begin{align}
	\expect{Y} \ge \expect{Y \mid U\ge \tfrac{m}{2}} \Prob{U \ge \tfrac{m}{2}}. 
	\label{eq:expbound}
\end{align}

By assumption, $p \ge \tfrac{c}{m}$ for a sufficiently large constant $c>0$.
Recalling that we have $|A\times B| = m^2$, it follows that $\expect{ U } = m^2 c / m  = c m$.
Since each edge becomes part of the target set independently with probability $p$, we can apply a standard Chernoff bound to show that $\Prob{ U \ge {m}/{2} } \ge 1 - 1/n^{\Omega(1)}$, and hence \eqref{eq:expbound} implies 
that the expected number of guesses is $\Omega(m)$.
The time bound follows since Alice can submit at most $2m$ guesses per round.
\end{proof}

\subsection{Lower Bounds for Information Dissemination} \label{sec:lbs}

In this section we show three different lower bounds.  Together, these show what properties cause poor performance in information dissemination protocols: in some graphs, high degree is the cause of poor performance (Theorem~\ref{thm:generalLBCAST}); in other graphs, poor connectivity is the cause of poor performance (Theorem~\ref{thm:generalConductance}).  And finally, we give a family of graphs where we can see the trade-off between $D$, $\Delta$, and $\phi_*$ (Theorem~\ref{thm:PUSHPULLOptimal}). We begin with a result showing that $\Omega(\Delta)$ is a lower bound: %

\begin{theorem} \label{thm:generalLBCAST}
	For any $\Delta \in (\Theta(1), \lceil{n}/{2}\rceil)$, there is an $n$-node network that has a weighted diameter of $O(\log n)$, and a maximum node degree $\Theta(\Delta)$, where any algorithm requires $\Omega(\Delta)$ rounds to solve local broadcast with constant probability.%
\end{theorem}
\begin{proof}
	Consider the network $H$ of $n$ nodes that consists of the guessing game gadget $G^\text{sym}(2\Delta,1,\Delta,P)$, where predicate $P$ returns an arbitrary singleton target set, combined with a constant degree regular expander \cite{hoory06} of $n - 2\Delta$ vertices (if any) of which any one node is connected to all the vertices on the left side of the gadget; all the edges of, and connected to the expander have latency $1$ and the latencies of the edges in the gadget are assigned as in Lemma \ref{lem:simulation}. Clearly, the weighted diameter of $H$ is $O(\log n)$ (diameter of the expander \cite{hoory06}). 
	By Lemma\ref{lem:singleton}, we know that any guessing game protocol on $\guessing(2\Delta,P(|T|=1))$ requires $\Omega(\Delta)$ rounds for the predicate that returns exactly $1$ pair as the target set. %
	Lemma \ref{lem:simulation} tells us that any gossip algorithm that solves local broadcast in $H$, must require $\Omega(\Delta)$ rounds.
\end{proof}

We next show that every local broadcast algorithm requires time at least $\Omega(1/\phi_* + \ell_*)$. Note that, we get this $\Omega(1/\phi_*)$ lower bound just for local broadcast and not information dissemination, which is in contrast to the results in the unweighted case. The following result is given in terms of the weight-$\ell$ conductance, for any $\ell$, and thus also holds for $\phi_*$ and $\ell_*$.  In the proof, we construct a network that corresponds to the bipartite guessing game graph with a target set where each edge is fast with probability $\phi_*$. That way, we obtain a network with critical weighted conductance $\Theta(\phi_*)$, hop diameter $O(1)$, and a weighted diameter of $O(\ell_*)$. The guessing game lower bound of Lemma \ref{lem:randomGuessing} tells us that the cost of information dissemination still depends on $\phi_*$. 

\begin{theorem} \label{thm:generalConductance}
	For any $\ell \in [1,n]$ and $\phi_\ell$ where $\Omega(\log(n)/n) \le \phi_\ell \le 1/2$, 
	there is a network of $2n$ nodes that has a weighted diameter $O(\ell)$ (w.h.p.), and critical weighted conductance $\Theta(\phi_\ell)$ (w.h.p.), such that any gossip algorithm requires $\Omega\left(({1}/{\phi_\ell})+ \ell\right)$ rounds for solving local broadcast in expectation. Also, solving local broadcast using push-pull requires $\Omega\left(({\log n}/{\phi_\ell}) + \ell\right)$ rounds in expectation.
\end{theorem}

\begin{proof}
  Our goal is to reduce the game $\guessing(2n,\random_{\phi_\ell})$ to local broadcast, hence we consider the $2n$-node graph $G(2n,\ell,n^2,\random_{\phi_\ell})$ as our guessing game gadget defined in Section \ref{sec:gadget}. Since we want to show the time bound of $t = \Omega\left(\tfrac{\log n}{\phi_\ell} + \ell\right)$ rounds (for push-pull), for the high latency edges we can use the value $n^2 \geq  \tfrac{\log n}{\phi_\ell} + \ell$ (as $\Omega(\log(n)/n) \leq \phi_\ell$ and $\ell \leq n$).
  
  We assign each cross edge latency $\ell$ independently with probability $\phi_\ell$ and latency $n^2$ with probability $1-\phi_\ell$. The fast cross edges have the same distribution as the target set implied by the predicate $\random_{\phi_\ell}$, which we have used to show a lower bound of $\Omega(\tfrac{1}{\phi_\ell})$ for general protocols on $\guessing(2n,\random_{\phi_\ell})$ in Lemma \ref{lem:randomGuessing}, and also a stronger lower bound of $\Omega(\tfrac{\log n}{\phi_\ell})$ for ``random guessing'' protocols, which choose a random edge for each vertex as their guesses. It is straightforward to see that push-pull gossip corresponds exactly to this random guessing game strategy. %
Applying Lemma \ref{lem:simulation}, this means that local broadcast requires in expectation $\Omega(\tfrac{1}{\phi_\ell})$ time for general algorithms and $\Omega(\tfrac{\log n}{\phi_\ell})$ time for push-pull. The additional term of $\Omega(\ell)$ in the theorem statement is required to actually send the broadcast over the latency $\ell$ edge once it is discovered.

Since each edge of $L \times R$ is assigned latency $\ell$ with probability $\phi_\ell=\Omega(\log(n)/n)$, it follows that each $u \in R$ is connected by a latency $\ell$ edge to some node in $L$ with high probability. Hence, the weighted diameter of $G(2n,\ell,n^2,\random_{\phi_\ell})$ is $O(\ell)$ with high probability.

In the remainder of the proof, we show that $G(2n,\ell,n^2,\random_{\phi_\ell})$ has a conductance of $\Theta(\phi_\ell)$ with high probability.
We point out that several previous works prove bounds on the network expansion (e.g., \cite{vadhan} and \cite{DBLP:conf/focs/AugustineP0RU15}). However, as these results were shown for random graphs, we cannot employ these results directly and thus need to adapt these proof techniques to show a conductance of $\Theta(\phi_\ell)$ for our guessing game gadget.

We assume that there is an integer-valued function $f=f(n)$, such that $\tfrac{f}{n} = \phi_\ell$, noting that this assumption does not change the asymptotic behavior of our bounds. 
For readability, we only consider $\ell=1$ and note that the extension to the general case is straightforward.
By construction, $G(2n,1,n^2,\random_{f/n})$ consists of edges with latencies $1$ or $n^2$ and we have \[
  \frac{ \phi_{n^2}}{n^2} \leq \frac{1}{n^2} \leq \frac{\phi_1}{1} ,%
\]
where the last inequality follows from the assumption $\phi_1\ge \Omega\left(\tfrac{\log n}{n}\right)$.
Thus, we know that $\phi_* = \phi_1$ and hence we need to prove $\phi_1 = \Theta(f/n)$.

Consider a set $S \subseteq L \cup R$ of at most $n$ vertices and let $l = |S \cap L|$ and $r = |S \cap R|$. %
We first assume that $l \ge r$, since the number of latency $1$ cross edges is symmetric for vertices in $L$ and $R$; subsequently, we will remove this assumption by a union bound argument.

For vertex sets $A$ and $B$, let $E_1(A,B)$ be the set of the (randomly sampled) latency $1$ edges in the cut $(A,B)$ and define $e_1(A,B) = |E_1(A,B)|$.
Given the set $S$, our goal is to show that many latency $1$ edges originating in $S \cap L$ have their other endpoint in $R \setminus S$, assuming that there are sufficiently many latency $1$ cross edges to begin with. 
In other words, we need to bound from above the probability of the event $e_1(S\cap L,S\cap R) \ge \Omega(f l)$ conditioned that there are sufficiently many latency $1$ cross edges.

\begin{claim}[Sufficiently many latency $1$ cross edges] \label{cl:concentration} 
There exist constants $c,c'>0$, such that events
\begin{align}
  \label{eq:concentration}
  \sLR & = \{ \forall S, |S|\le n\colon (e_1(S\cap L,R)\!\ge\! c f l) \wedge  (e_1(S\cap R,L)\!\ge\! c f r) \},\\
  \label{eq:upperConcentration}
  \sLR^- & = \{ \forall S, |S|\le n\colon (e_1(S\cap L,R)\!\le\! c' f l) \wedge  (e_1(S\cap R,L)\!\le\! c' f r) \}
\end{align}
occur with high probability.
\end{claim}
\begin{proof}
According to the construction of $G(2n,1,n^2,\random_{f/n})$, the latency $1$ cross edges are chosen independently each with probability $f/n$. 
Note that each cross edges is assigned latency $1$ independently with probability $f/n = \Omega(\tfrac{\log n}{n})$. 
Thus, for each node $v$, the expected number of cross edges is $f = \Omega(\log n)$ and, by a standard Chernoff bound, we know that the number of latency $1$ cross edges to $v$ is in $[c_1 f, c_2 f]$ with high probability, for suitable constants $c_2\ge c_1>0$.
After taking a union bound over all nodes in $V(G)$, we can conclude that the claim holds for any set $S \subseteq V(G)$.
\end{proof}

Conditioning on $\sLR$ is equivalent with choosing a subset of (at least) $c f l$ edges among all possible edges in the cut $E_1(S\cap L,R)$ uniformly at random and assigning them latency $1$.
Consider an edge $(v,u) \in E_1(S\cap L,R)$.
It follows that $u \in S \cap R$ (and hence $(v,u) \in E_1(S\cap L, S \cap R)$), with probability $\frac{r}{n}$ and we need to exclude the event 
\[
  \bad(S) =  \{e_1(S \cap L,S\cap R) \ge \tfrac{4}{5} c f l\}, 
\]
  for all ${c f l \choose \frac{4}{5} c f l}$ subsets of latency $1$ edges incident to vertices in $S \cap L$.
In addition, we need to bound the probability that $\bad(S)$ happens, for $S$ chosen in any of the ${n \choose l}$ ways of choosing $S$ that satisfy $|S\cap L|=l$.
\begin{claim} \label{cl:bad}
  $\Prob{\exists S\colon \bad(S) \suchthat \sLR }  \le n^{-\Omega(1)}.$
\end{claim}
\begin{proof}[Proof of claim.]
Combining the above observations, we get
\begin{align} %
  \Prob{\exists S\colon \bad(S) \suchthat \sLR } & \le {n \choose l} { c f l \choose \tfrac{4}{5} c f l } \left( \frac{r}{n} \right)^{ \frac{4}{5} c f l}. \label{eq:bound}
  \intertext{First, we assume that $r$ and $l$ are both large, i.e., $l \ge r \ge c' n$, for a sufficiently small positive constant $c'<\tfrac{4}{5e}$. Then, we can apply Stirling's approximation of the form ${m \choose k} \approx 2^{m \cdot\text{H}_2(\frac{k}{m})}$, where $\text{H}_2(x)=-x \log_2(x) - (1-x)\log(1-x)$ is the binary entropy function. Thus, for sufficiently large $n$, we get
}
  \Prob{\exists S\colon \bad(S) \suchthat \sLR } 
  &\le 2^{n\cdot\text{H}_2(\frac{l}{n}) + c f l  \text{H}_2\left(\frac{4}{5}\right)}\cdot\left(\frac{r}{n}\right)^{\tfrac{4}{5} c f l}\notag \\
  &\le 2^{n + c f l \cdot \text{H}_2\left(\frac{4}{5}\right) - \frac{4}{5} c f l}, \label{eq:bound2}
\intertext{
  where, to derive the second inequality, we have used the facts that $\text{H}_2(\tfrac{l}{n}) \le 1$ and $\tfrac{r}{n} \le \tfrac{1}{2}$, since $r + l \le n$ and $r \le l$. 
  By the premise of the theorem, $f = \Omega(\log n)$, which implies $c\ f\ l = \Omega(n \log n)$.
  Together with the fact that $\text{H}_2(\tfrac{4}{5}) < \tfrac{4}{5}$, this means that the term $(-\frac{4}{5} c\ f\ l)$ dominates in the exponent of \eqref{eq:bound2} and hence
}
  \Prob{\exists S\colon \bad(S) \suchthat \sLR } 
  &\le 2^{- \Theta(n \log n)}.\notag 
  \intertext{Next, we consider the case where $r \le l < c' n$. Applying the upper bound of the form ${m \choose k} \le \left(\tfrac{em}{k}\right)^k$ to \eqref{eq:bound}, tells us that
}
  \Prob{\exists S\colon \bad(S) \suchthat \sLR } 
  &\le \left(\frac{e n}{l}\right)^l \left(\frac{5e\ r}{4n}\right)^{\frac{4}{5} c f l}\notag \\
  &\le \left(\frac{e n}{l}\right)^l \left(\frac{5c' e}{4}\right)^{\frac{4}{5} c f l},\notag 
  \intertext{since $r < c' n$. We get}
  \Prob{\exists S\colon \bad(S) \suchthat \sLR } 
  &\le \exp\left(l \left(1+ \log n - \log l + \frac{4}{5} c f l \log\left(\frac{5}{4}c' e\right)\right)\right)\notag \\
  &\le \exp\left(l \left(1+ \log n + \frac{4}{5} c f l \log\left(\frac{5}{4}c' e\right)\right)\right).\notag 
\notag
\end{align}

  By assumption, $c'<\tfrac{4}{5e}$ and hence the term $ \frac{4}{5} c\ f\ l \log\left(\frac{5}{4}c' e\right)$ in the exponent is negative. Moreover, recall that $f = \Omega(\log n)$ and thus we can assume that $\frac{4}{5} c\ f\ l \log\left(\frac{5}{4}c' e\right) \le - c'' \log n$, for a sufficiently large constant $c''>0$. This term dominates the other terms in the exponent, thereby completing the proof of the claim.
\end{proof}
Considering that $l \ge \tfrac{|S|}{2}$, the above bound implies that at least $\lfloor\tfrac{c f |S|}{10}\rfloor$ latency $1$ edges incident to $S$ are connected to nodes outside in $S$, with probability at least $1 - n^{-\Omega(1)}$.
Taking a union bound over all possible choices for the values of $l$ and $r$ adhering to $r\le l$ and $r + l \le n = \frac{|V(G)|}{2}$, shows that
\begin{align*} 
  \Prob{\forall S, |S|\le {n}\colon e_1(S\cap L,R\setminus S) \ge \tfrac{c f|S|}{10} \suchthat \sLR } \ge 1 - n^{-\Omega(1)}.
\end{align*}

Observe that the latency $1$ cross edges are constructed symmetrically for the left and right side of the bipartite graph $G$ and thus we can apply the above argument in a similar manner for a set $S$ where $r> l$, conditioned on $e_1(S\cap R,L)\!\ge\! c f l$.
Thus, we can conclude that %
 \begin{align*} \label{eq:probBound}
   \Prob{\forall S, |S|\le n\colon e_1(S,V(G)\setminus\! S) \ge \tfrac{c f|S|}{10} \suchthat \sLR } \ge 1 - n^{-\Omega(1)}.
 \end{align*}

We can remove the conditioning in the above equation by virtue of Claim~\eqref{cl:concentration}, since
 \begin{align*}
   \Prob{\forall S, |S|\le n\colon e_1(S,V(G)\setminus\! S) \ge \tfrac{c f|S|}{10} } 
   &\ge \Prob{\forall S, |S|\le n\colon e_1(S,V(G)\setminus\! S) \ge \tfrac{c f|S|}{10} \suchthat \sLR }  \Prob{ \sLR }  \\
   &\ge 1 - n^{-\Omega(1)}.
 \end{align*}

To upper bound $\vol(S)$ for any set $S$, we take into account the $n$ cross edges  of each node in $S$. 
Also, if $v \in L$, then we need to account for the $n-1$ incident clique edges of $v$, yielding $\vol(S) \le 2|S|n$.
Considering the upper bound on the number of latency $1$ cross edges given by \eqref{eq:upperConcentration}, we have
 \[
   \phi_* = \min_S \phi_1(S) = \min_S \frac{e_1(S,V(G)\setminus S)}{\vol(S)} \]
\[   \ge \min_S\frac{c f |S|}{20|S|n} = \Omega(\tfrac{f}{n}),
 \]

where the inequality is true with high probability.
To see that this bound is tight, observe that $\phi_* \le \phi_1(L)$.  
By \eqref{eq:concentration} and \eqref{eq:upperConcentration}, we know that $e_1(L,R) = \Theta(fn)$ and hence $\phi_1(L) = \Theta\left(\tfrac{f}{n}\right)$ with high probability, as required. This completes the proof of Theorem \ref{thm:generalConductance}.
\end{proof}

Finally, we give a family of graphs that illustrate the trade-off among the parameters. Intuitively, when the edge latencies are larger, it makes sense to search for the best possible path and the lower bound is $\Omega(D + \Delta)$; when the edge latencies are smaller, then we can simply rely on connectivity and the lower bound is $\Omega(\ell/\phi_{\ell})$. 
Note that, we can individually obtain a lower bound of $\Omega((\ell/\phi)\log n)$, using the technique in \cite{Chierichetti:2010:ATB:1806689.1806745} where we show that there exists a graph with diameter $(\ell/\phi)\log n$. Unlike here, that lower bound is simply $D$.

\begin{theorem} \label{thm:PUSHPULLOptimal}
	For a given $\alpha \in [\Omega({1}/{n}),O(1)]$ and any integer $\ell \in [1,O( n^2\alpha^2)]$, there is a class of networks of $2n$ nodes, critical weighted conductance $\phi_* = \phi_\ell = \Theta(\alpha)$, maximum degree $\Delta = \Theta(\alpha n)$, and weighted diameter $D = \Theta({1}/{\phi_\ell})$, such that any gossip algorithm that solves broadcast with at least constant probability, requires 
	$\Omega\left(\min\left\{ \Delta + \D, {\ell}/{\phi_\ell} \right\}\right)$ 
	rounds.
\end{theorem}

\begin{proof}
	We create a network $G$ consisting of a series of $k$ node layers $V_1,\dots,V_{k}$ that are wired together as a ring, using the guessing game gadgets introduced above. 
	We define $k=\tfrac{2}{c\alpha}$ where $c=\left(\tfrac{3}{4} + \tfrac{1}{4} \sqrt{9-\tfrac{8}{n\alpha}}\right)$. This implies that $1 \le c < 3/2$ as $\alpha \in [\Omega({1}/{n}),O(1)]$.  Each layer consists of $s = cn\alpha$ nodes.
	As it does not change our asymptotic bounds, we simplify the notation by assuming that 
	$2/c\alpha$ and $c n \alpha$ are integers.

\begin{figure}[h]
	\centering
	\includegraphics[width=0.48\textwidth , height=0.40\textwidth]{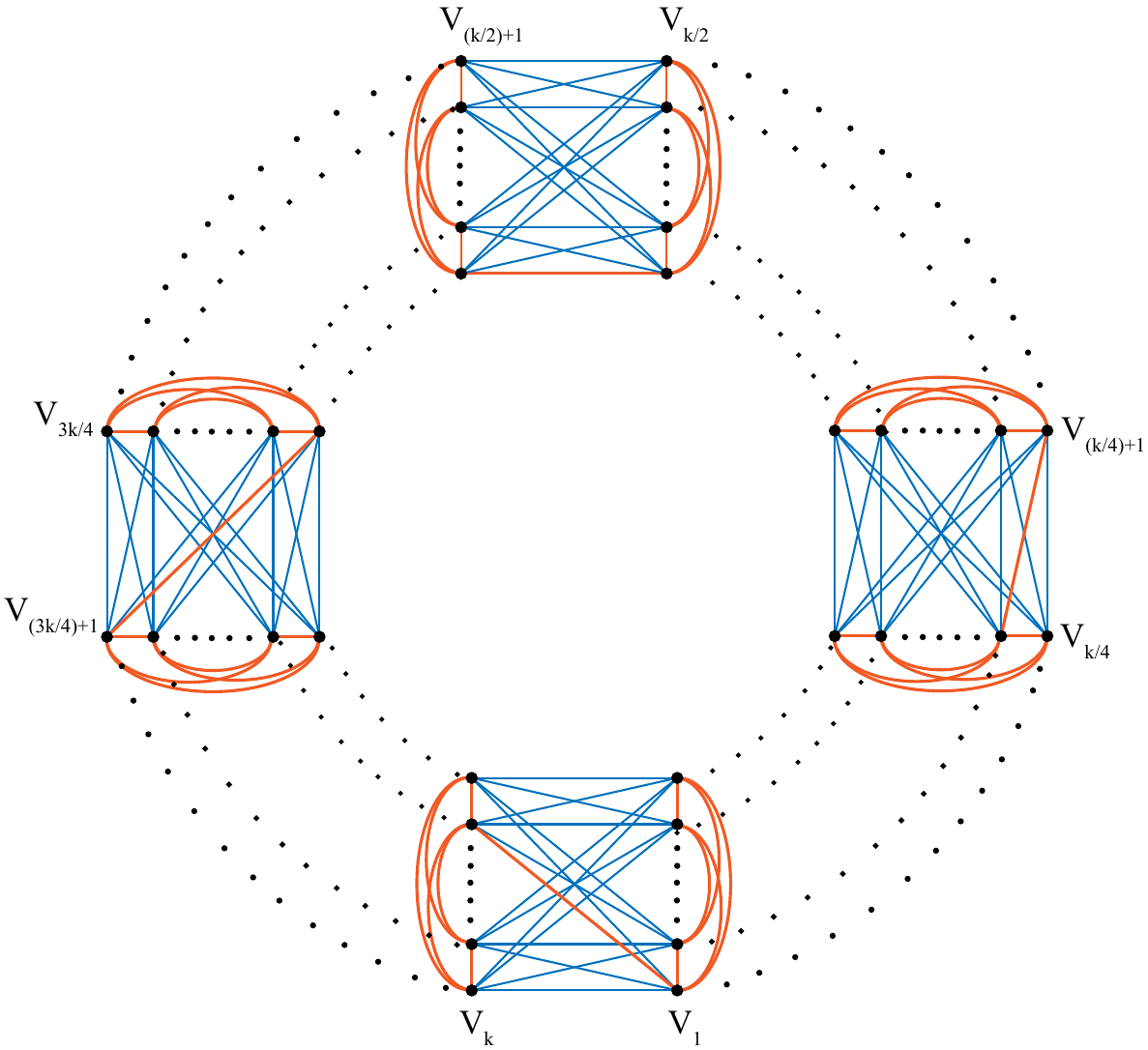}
	\caption{Guessing Game Gadgets wired together as a ring.}
	\label{fig:ring}
\end{figure}
	
	\noindent For each pair $V_i$ and $V_{(i+1)\text{ mod }k}$ ($0 \le i \le k-1$), we construct the symmetric guessing game gadget $G^\text{sym}(2cn\alpha,1,\ell,P)$ (in Section \ref{sec:gadget}), for simulating a gossip algorithm to solve the game  $\guessing(2cn\alpha,P(|T|=1))$. 
	That is, we create a complete bipartite graph on $V_i$ and $V_{(i+1)\text{ mod }k}$ and form cliques on $V_i$ and $V_{(i+1)\text{ mod }k}$ (see Figure \ref{fig:ring}). %
	We assign latency $\ell$ to every cross edge between $V_i$ and $V_{(i+1)\text{ mod }k}$, except for a uniformly at random chosen edge that forms the singleton target set, which we assign latency $1$. %
	Observe that the weight-$j$ conductance $\phi_j$ %
	cannot be maximal for any $j$ other than $1$ or $\ell$. %
 
\begin{observation} \label{ob:degree}
Let $s = cn\alpha$.
Graph $G$ is $(3s-1)$-regular.
\end{observation}
\begin{proof}
For a layer $V_i$, we call $V_{(i-1)\text{ mod }k}$ the \emph{predecessor layer} and $V_{(i+1)\text{ mod }k}$ the \emph{successor layer}.
The size of a layer is $s = cn\alpha$.
Each node has $2s$ edges to its neighbors in the predecessor resp.\ successor layer and $s-1$ edges to nodes in its own layer.
This means that $G$ is a $(3s-1)$-regular graph.
\end{proof}
	
	We define a cut $C$ that divides the ring into two equal halves such that none of the internal clique edges are cut edges. By a slight abuse of notation, we also use $C$ to denote the set of vertices present in the smaller side of the partition created by the cut $C$ (ties broken arbitrarily).

	\begin{lemma} \label{lem:alphaphi}
		$\phi_\ell(C) = \alpha$.
	\end{lemma}
	
	\begin{proof}
Since $C$ partitions $G$ into two sets of identical size, the volume can be determined by considering either partition of size $n$, thus we focus on the node set $C$. Also, by Observation \ref{ob:degree} we know that $G$ is $(3s-1)$-regular. %
The volume of $C$ can be calculated to be $n(3cn\alpha-1)$. The number of cut edges of latency $\le \ell$ is $2(cn\alpha)^2$ (by the construction of $C$). According to Definition \ref{def:dweightedcond}, the $\ell$-weight conductance is given by
$
  \phi_{\ell}(C) = \frac{2(cn\alpha)^2}{n(3cn\alpha-1)}.
$
By plugging in the value of $c$, we can verify that $\phi_\ell(C)$ is exactly equal to $\alpha$.
\end{proof}
	
	Using the conductance bound of Lemma \ref{lem:alphaphi} for cut $C$, we know that $\phi_\ell \le \alpha$. In the proof of the next lemma, we show that $\phi_\ell = \Omega(\alpha)$.

	\begin{lemma} \label{lem:correctphi}
		The weight-$\ell$ conductance of the constructed ring network is $\phi_\ell = \Theta(\alpha)$.
	\end{lemma}

\begin{proof}
By Lemma \ref{lem:alphaphi}, we know that $\phi_\ell \le \alpha$ as the actual graph conductance is always $\le$ to any cut conductance. We will now show $\phi_\ell = \Omega(\alpha)$ as well.

By Observation \ref{ob:degree} we know that $G$ is $(3s-1)$-regular and therefore for a set of nodes $U$ the volume $\vol(U)$ is exactly equal to $(3s-1)|U|$. This clearly implies that for any two sets $U$ and $V$, $\vol(U) \le \vol(V)$ if and only if $ |U|\le |V|$. 

Now, consider an arbitrary cut $(U,V(G)\setminus U)$ of $G$ and suppose that $U$ contains at most half of the nodes of $G$, i.e., $|U| \le n$, since $G$ has $2n$ nodes.
If there are at least $\Theta(s^2)$ cut edges, then, using the fact that $|U| \le n$, we get 
\[
\phi_\ell(U) \ge \Theta(s^2/s|U|) = \Theta(s/|U|) \ge \Theta(s/n) \ge \Theta(\alpha),
\]
and we are done.
In the remainder of the proof, we will show that there are $\Theta(s^2)$ cut edges. 
We distinguish two cases:\\
{\boldmath 1. $|U|\ge 3s/4$:}\\
We classify each node in $U$ either as \emph{good} if it has at least $s/4$ adjacent edges across the cut $(U,V\setminus U)$ and as \emph{bad} otherwise.  
Thus, our goal is to identify $\Theta(s)$ good nodes, which in turn implies $\Theta(s^2)$ cut edges.

Let $S$ be an arbitrary subset of $3s/4$ nodes in $U$. If all nodes in $S$ are good, we are done.  Otherwise, let $x \in S$ be a bad node. It is important to note that the following properties are true for every bad node:
\begin{compactitem}
\item[(a)] Node $x$ is in a layer in $G$ which contains at least $3s/4$ nodes inside $U$.  
\item[(b)] The successor layer from $x$ has at least $3s/4$ nodes inside $U$.  
\end{compactitem}
To see why (a) holds, assume that it was not true. Then, $x$ would have at least $s/4$ neighbors in its own layer across the cut, contradicting the assumption that $x$ is bad.
Similarly, if (b) was false, $x$ would be connected to at least $s/4$ nodes in the successor layer outside $U$. (This is true of the predecessor layer too.)

Let $A$ be the successor layer to the layer containing $x$. We now run the following procedure:
\begin{compactenum}
\item[(1)] \textbf{Invariant:} $A$ contains at least $3s/4$ nodes in $U$.
\item[(2)] If at least half of the nodes in $A$ are good, we are done. Terminate and claim $\Theta(s^2)$ cut edges.
\item[(3)] Otherwise, let $y$ be a bad node in $A$.
\item[(4)] Let $A'$ be the successor layer of the layer $A$. Then, start again at Step (1) with $A=A'$ and $y=x$.
  From the assertion (b),  $A'$ contains at least $3s/4$ nodes in $S$.
\end{compactenum}

If this procedure ever terminates in Step (2), we are done.  
Otherwise, it continues around until every layer has been explored.  In that case, the invariant implies that every layer contains at least $3s/4$ nodes in $U$.  This implies that $> 1/2$ of the nodes of $G$ are in $U$, which contradicts the choice of $U$.  
Thus, the procedure does terminate, which means there must be at least $\Theta(s^2)$ cut edges, implying $\phi_\ell \ge \alpha$.  \\
{\boldmath 2. $|U| < 3s/4$:}\\
Let $m$ be the number of nodes in $U$. Since $G$ is $(3s-1)$-regular, the volume of $U$ is $m(3s-1)$.  Each node in $U$ now contains at least $s/4$ neighbors outside of $U$ (since it has $\ge s$ neighbors and there are only $< 3s/4$ other nodes in $U$), so the cut size is at least $sm/4$.  Thus, the conductance of this graph $\phi_\ell \ge \frac{(sm/4)}{m(3s-1)} = \Omega(1) \ge \Theta(\alpha$).

Since, $\phi_\ell \le \alpha$ and $\phi_\ell \ge \Theta(\alpha)$, it is clearly the case that $\phi_\ell = \Theta(\alpha)$, which is what we wanted to prove.%
\end{proof}

	Combining Lemmas~\ref{lem:alphaphi} and \ref{lem:correctphi} (and again using cut $C$), we argue that the critical latency is $\ell$.
	
	\begin{lemma} \label{lem:phibound}
		For any $\ell \le O(cn\alpha)^2$, $\phi_* =\phi_\ell=\Theta(\alpha)$.
	\end{lemma}

\begin{proof}
To prove that $\phi_*$ is in fact $ \phi_\ell$, which by Lemma \ref{lem:correctphi} is $\Theta(\alpha)$, we need to show that $(\phi_\ell/ \ell) \ge (\phi_1/1) = \phi_1$. 
To this end, let us consider the cut $C$ defined above. We will show that  
$\frac{\phi_\ell}{\ell } \ge \phi_1(C) \ge \phi_1$, and since weight-$j$ conductance $\phi_j$ (cf. Definition \ref{def:dweightedcond}), cannot be maximal for any $j$ other than $1$ or $\ell$, we get $\phi_* = \phi_\ell$.

There are two latency $1$ cross edges in the cut $C$ and the volume of $C$ can be calculated as in the proof of Lemma \ref{lem:alphaphi} to be $n(3cn\alpha-1)$.
Thus, we need to show that
\[
  \frac{\phi_\ell}{\ell} = \frac{\Theta(\alpha)}{\ell} \ge \frac{2}{(3cn\alpha - 1)n}.
\]
As $c$ is constant, the 
inequality is true as long as $\ell = O(\alpha^2 n^2)$, which is ensured by the premise of the theorem. 
\end{proof}

\noindent The weighted diameter of the network $D = \Theta(k/2)$, since each pair of adjacent node layers is connected by a latency $1$ edge and, internally, each layer forms a latency $1$ clique. Using the fact that $c \in [1, \frac{3}{2})$, it can be shown that $ (2/3\alpha) < D \le (1/\alpha)$, implying that $D = \Theta(1/\phi_\ell)$ (by lemma \ref{lem:correctphi}).
	
Now, consider a source node in layer $V_1$ that initiates the broadcast of a rumor. Each node can either spend time in finding the required fast edge (which we assume can be done in parallel) or, instead, it can instantly use an edge of latency $\ell$ to forward the rumor.
	Lemma \ref{lem:singleton} tells us that finding the single latency $1$ cross edge  with constant probability, for the guessing game gadget corresponding to any pair of node layers, requires $\Omega(\Delta)$ rounds, and then forwarding the rumor takes $\Omega(D)$ additional rounds.
	Alternatively, the algorithm can forward the rumor along latency $\ell$ edges across node layers and spread the rumor using the latency $1$ edges within each clique. 
	It follows that the required time for broadcast is 
	$
	\Omega\left(\min\left\{ \Delta + D, {\ell}/{\phi_\ell} \right\}\right).
	$
\end{proof}

We obtain the following corollary that gives a lower bound on information dissemination in terms of $\phi_{avg}$, either by a similar analysis as above, or by the application of Theorem \ref{thm:conductance}.

\begin{corollary} \label{coro:avglower}
For a given $\alpha \in [\Omega({1}/{n}),O(1)]$ and any integer $\ell \in [1,O( n^2\alpha^2)]$, there is a class of networks of $2n$ nodes, average weighted conductance $\phi_{avg} = \Theta(\alpha/\ell)$, maximum degree $\Delta = \Theta(\alpha n)$, and weighted diameter $D = \Theta(1/{\ell}{\phi_{avg}})$, such that any gossip algorithm that solves broadcast with at least constant probability, requires 
  $\Omega\left(\min\left\{ \Delta + \D, 1/{\phi_{avg}} \right\}\right)$ 
rounds.
\end{corollary}

\begin{proof}
Observe that in the given graph, there exists edges with latency either $1$ or $\ell$, and as such the number of non-empty latency classes here is $2$. Now,  Theorem \ref{thm:conductance} reduces to ${\phi_*}/{2\ell_*} < \phi_{avg} < 2{\phi_*}/{\ell_*}$. This implies that for this case $\phi_{avg}= \Theta({\phi_*}/{\ell_*})$. Alternatively, $\phi_* = \ell \phi_{avg}$ (as in this case $\ell=\ell*$). Replacing this value of $\phi_*$ in  Theorem \ref{thm:PUSHPULLOptimal} gives us the above required corollary.
\end{proof}
\section{Algorithms for Known Latencies}
\label{sec:known}
In this section, we discuss the case where each node knows the latencies of all its adjacent edges. Later, in Section \ref{sec:upper_bounds}, we provide upper bounds for the case where nodes are not aware of the edge latencies.
For this section only, we focus on the problem of all-to-all information dissemination (instead of one-to-all information dissemination), as it will simplify certain issues to solve the seemingly harder problem.  

Here, we provide two different solutions to the problem of all-to-all information dissemination. In Section \ref{spanner_algo} we provide a spanner based randomized algorithm that solves all-to-all information dissemination in $O(D\log^3n)$ rounds w.h.p., whereas in Section \ref{app:altalgo} we provide a pattern based deterministic solution for all-to-all broadcast taking $O(D\log^2n\log D)$ rounds. The additional $\log D$ factor (instead of $\log n$) makes the pattern based algorithm unsuitable for graphs with large diameters. Note that, for either algorithm, we assume that messages can be of polynomial size (in $n$). 

\para{All-to-all information dissemination} {%
Initially, each node begins with a source message and, when the protocol terminates, every node must have received all other source messages.}

(Of course, all-to-all information dissemination also solves one-to-all information dissemination. Furthermore, most one-to-all information dissemination algorithms can be used to solve all-to-all information dissemination by using them to collect and disseminate data.)

\subsection{Spanner Broadcast Algorithm}
\label{spanner_algo}

In Section \ref{spanner_algo}, we use the fact that nodes know a polynomial upper bound on the network size (and this is the only place in this paper where we rely on that assumption). 
When edge latencies are known, the spanner algorithm (described below) solves all-to-all information dissemination in $O(D\log^3n)$ which differs from the trivial lower bound of $\Omega(D)$ by only polylog factors.  %

\subsubsection{Preliminaries}

We initially assume that the weighted diameter ($D$) is known to all nodes; later (in Section \ref{sec:diameter}), we do away with the assumption via a guess-and-double technique. It is assumed w.l.o.g. that every edge has latency $\leq D$: clearly we do not want to use any edges with latency $> D$.

\para{Local broadcast}
\label{back_algo}
An important building block of our algorithms is \emph{local broadcast}. For unweighted graphs, the (randomized) {Superstep} algorithm by Censor-Hillel et al.~\cite{Censor-Hillel:2012:GCP:2213977.2214064} and the Deterministic Tree Gossip (DTG) algorithm by Haeupler~\cite{Haeupler:2013:SFD:2627817.2627868} solve this problem.  We make use of the DTG algorithm, which runs in $O(\log^2{n})$ rounds on unweighted graphs. See \cite{Haeupler:2013:SFD:2627817.2627868} %
for details. %
Observe that for the unweighted case, if any algorithm solves local broadcast in $O(t)$ rounds, it obtains a $t$-spanner as a direct consequence, which thereafter can be used for propagating information. However, for graphs with latencies, just solving local broadcast might take $O(D)$ time, resulting in a $O(D)$-spanner (and leading to an $O(D^2)$ solution for information dissemination). Recall that a subgraph $S = (V, E')$ of a graph $G = (V, E)$ is called an $\alpha$-spanner if any two nodes $u, v$ with distance $\ell$ in $G$ have distance at most $\alpha \ell$ in $S$.

For weighted graphs, we are mainly interested in the \emph{$\ell$-local broadcast problem} in which each node disseminates some information to all its neighbors that are connected to it by  edges of latency $\leq \ell$.  While DTG assumes edges to be unweighted (uniform weight), 
we can execute the same protocol in a graph with non-uniform latencies simply by ignoring all edges with a latency larger than $\ell$ and simulating $1$ round of the DTG protocol as $\ell$ rounds in our network. We refer to this protocol as the $\ell$-DTG protocol.  
It follows immediately that within $O(\ell \log^2{n})$ time, the $\ell$-DTG protocol ensures that each node has disseminated the information to all its neighbors connected to it with edges of latency $\leq \ell$. 
Note that we can trivially solve the all-to-all information dissemination problem in $O(D^2 \log^2 n)$ time using $\ell$-DTG protocol (if $D$ were known) by simply repeating it $D$ times with $\ell = D$.  

The challenge now, given the restriction that finding neighbors by a direct edge might be costly, is to somehow find sufficiently short paths to all of them. %
We show here that with sufficient exploration of the local neighborhood up to $O(\log n)$ steps and using only favorable weights, we are able to obtain a global spanner.
An intermediate goal of our algorithm is to construct an $O(\log{n})$-spanner and to obtain an orientation of the edges such that each node has a small, i.e., $O(\log n)$, out-degree.\footnote{It is clearly impossible to guarantee small degree in an undirected sense, for example, if the original graph is a star.} Once we have such a structure, we achieve all-to-all information dissemination by using a flooding algorithm that repeatedly activates the out-edges in round-robin order.

\subsubsection{Spanner Construction Procedure}   
\label{sec:spanner}

In a seminal work, Baswana and Sen \cite{baswana} provide a spanner construction algorithm for weighted graphs (where weights did not correspond to latency) in the $\mathcal{LOCAL}$ model of communication. As our goal here is to find a low stretch, low out-degree spanner,  we modify the algorithm of \cite{baswana} by carefully associating a direction with every edge that is added to a spanner such that each node has w.h.p. $O(\log{n})$ out-degree. To deal with latencies, we choose to locally simulate the algorithm on individual nodes after obtaining the $\log n$-hop neighborhood information by using the $\ell$-DTG protocol. We show that this $\log n$-hop neighborhood information is sufficient for obtaining the required spanner. %
The algorithm in \cite{baswana} also assumes distinct edge weights. We can ensure this by using the unique node IDs to break ties.

Each node $v$ executes a set of rules for adding edges (explained below) and each time one of these rules is triggered, $v$ adds some of its incident edges to the spanner while assigning them as \emph{outgoing} direction. 
This way, we obtain a low stretch spanner (undirected stretch) where nodes also have a low out-degree, which we leverage in the subsequent phases of our algorithm.

For a given parameter $k$, the algorithm  computes a $(2k-1)$-spanner by performing $k$ iterations.
At the beginning of the $i$-th iteration, for $1 \le i \le k-1$, every node that was a cluster center in the previous iteration, chooses to become an active cluster with probability $\hat{n}^{-1/k}$, for some $n \le \hat{n} \le \poly(n)$; note that for $i=1$, every node counts as a previously active center.
Then, every active center $c$ broadcasts this information to all cluster members.
As a cluster grows by at most $1$ hop in each round, this message needs to be disseminated throughout the $i$-neighborhood of $c$.\footnote{By slight abuse of notation, we use $c$ to denote cluster centres and the cluster itself when the distinction is clear from the context.}
Then, every cluster member broadcasts its membership information to all its neighbors to ensure that every node is aware of its adjacent active clusters.
For adding edges to the spanner, nodes also remember its set of incident clusters $C_{i-1}$ that were active in iteration $i-1$.
With this information in hand, every node $u$ adds some of its incident edges to its set of spanner edges $H_u$, and also (permanently) discards some edges, as follows:\\
\emph{(Rule 1)} If none of $u$'s adjacent clusters in $C_{i-1}$ were sampled in iteration $i$, then $u$ adds its least weight edge to cluster $c$ as an outgoing edge to $H_u$ and discards all other edges to nodes in $c$, for every $c \in C_{i-1}$ .\\
\emph{(Rule 2)} If $u$ has active adjacent clusters, then $u$ will add the edge $e_v$ to some cluster $c$ with the minimum weight among all these clusters and, for each adjacent cluster $c' \in C_{i-1}$ that has a weight less than $e_v$, node $u$ also adds one outgoing edge to the respective node in $c'$. All other edges from $v$ to nodes in clusters $c$ and $c'$ are discarded.

In the $k$-th iteration, every vertex $v$ adds the least weight edge to each adjacent cluster in $C_{k-1}$ to $H_v$.

We first show that the size of the obtained spanner does not increase significantly when running the algorithm of \cite{baswana} with an estimate of $n$ (namely $\hat{n}$).

\begin{lemma}\label{lem:spannerClassic}
	Consider a synchronous network of $n$ nodes where nodes know only $\hat{n}$,  where $n \le \hat{n} \le n^c$, for some constant $c\ge 1$.
	For any $k\ge c$, there's a distributed algorithm (based on \cite{baswana}) that computes a spanner and terminates in $O(k)$ rounds in the $\mathcal{LOCAL}$ model %
	and each node's out-degree is $O(n^{{c}/{k}}\log{n})$ w.h.p. 
\end{lemma}

\begin{proof}
Note that the running time of the algorithm is $O(k^2)$ rounds if used with a restricted message size of $O(\log n)$. 
Inspecting the algorithm reveals that the computation at each node only depends on its $k$-hop neighborhood in the graph. Also, because the decision to remove an edge $(u,v)$ can be taken by either node $u$ or $v$, each node needs to simulate the running of the algorithm at all its neighbors (to know when to remove the edge $(u,v)$ from consideration) and hence we can simulate the execution of the algorithm locally by first collecting this information regarding $(k+1)$-hop neighborhood in $k+1$ rounds in the $\mathcal{LOCAL}$ model. 

We now analyze the difference when running the algorithm with $\hat{n}$ instead of $n$. 
First, we observe that sampling clusters with probability $\hat{n}^{(-1/k)}$ does not affect the stretch guarantee.
For the sake of our analysis we assume that the spanner is directed: we count every incident edge of $v$ that it adds to its set of spanner edges $H_v$ %
as an \emph{outgoing edge} of $v$.
The degree bound will follow by showing an upper bound on the number of outgoing edges of each node.

Consider any iteration $i$ in Phase~1 of the algorithm, i.e., $1\le i < k$.
We call a cluster \emph{sampled in iteration $i$} if it is among the sampled clusters in all iterations $1,\dots,i$.
Every cluster that was sampled in the previous iteration is sampled again with probability $\hat{n}^{-1/k}$. (In the very first iteration, every node counts as a previously sampled cluster.) 
To bound the number of edges that contribute to the out-degree of a node $v$, we consider the clusters adjacent to $v$ that were sampled in iteration $i-1$ and order them as $c_1,\dots,c_{q}$ in increasing order of the weight of their least weight edge incident to $v$.

Let $A_{i}$ be the event that $v$ adds at least $l$ edges to its outdegree in iteration $i$.
Note that $A_{i}$ occurs if and only if (1) none of the clusters $c_1,\dots c_l$ is sampled in iteration $i$ and (2) there are at least $l$ active clusters in iteration $i-1$.
By the description of Phase~1 (first $k-1$ iterations) of the algorithm, we only add an edge from $v$ to a node in cluster $c_j$ in iteration $i$ if $A_{i}$ does not happen.
We have $\Prob{ A_{i} } \le (1 - n^{-c/k})^l$ and taking a union bound over the first $k-1$ iterations and over all $n$ nodes, it follows that the probability of any node adding more than $l$ edges to the spanner in any of the first $k-1$ iterations is at most
  $\exp(-n^{-c/k}l + \log k + \log n)$.
  By choosing $l \ge \Omega(n^{1/k}(\log n + \log k))$, this probability is $\le n^{-\Omega(1)}$ as required.

In Phase~2 (final iteration), every vertex $u$ adds a least weight (outgoing) edge to every cluster that was sampled in iteration $k-1$. 
Let $X_v$ be the indicator random variable that vertex $v$ is the center of a cluster sampled in iteration $k-1$ that is incident to $u$.
We have 
\[
  \Prob{X_v} \le n^{-\frac{c(k-1)}{k}} = n^{-c + \frac{c}{k}}.
\]
Setting $X = \sum_{v: (u,v) \in G}X_v$, it follows that 
\[
  \expect{X} \le n^{1-c + \frac{c}{k}}\le n^{\frac{c}{k}}, 
\]
  since $c\ge 1$.
Since each cluster is sampled independently all $X_v$ are independent, we can apply a standard Chernoff bound to show that, for some sufficiently large constant $c_1$ depending on $c$, it holds that 
\[
  \Prob{X \ge c_1 n^{\frac{c}{k}}\log n} \le e^{-\Theta(n^{c/k}\log n)} \le n^{-\Omega(1)}. 
\]
By taking a union bound over all vertices, we can see
the number of edges that each vertex adds to the spanner in Phase~2 is at most $O(n^{\frac{c}{k}}\log n)$ with high probability. 
Combining this with the bound that we have derived for Phase~1 completes the proof.
\end{proof}

\begin{theorem}
	There is an $O(D\log^3 n)$ time algorithm $\mathcal{A}$ in the gossip model that yields an $O(\log n)$-spanner that has $O(n\log{n})$ edges (w.h.p.).
	Moreover, $\mathcal{A}$ also computes an edge orientation that guarantees that each node has an out-degree of $O(\log{n})$ (w.h.p.).
\end{theorem}

\begin{proof}
	
	To convert the classic synchronous algorithm for the local model assumed in Lemma \ref{lem:spannerClassic} to an algorithm that works in the gossip model with latencies, we use %
	the $\ell$-DTG protocol %
	and simulate each of the $k = \log n$ iterations of the spanner algorithm by first discovering the $\log n$-hop neighborhood. The neighborhood discovery takes $O(D\log^3{n})$ rounds in our model and then all computations are done locally.
\end{proof}

\subsubsection{Broadcasting on the Directed Spanner}
\label{sec:roundrobin}
To broadcast on this directed spanner we use the RR broadcast algorithm, %
which is a deterministic round-robin-style exchange of information among nodes.
Each node sends all the rumors known to it to all its $1$-hop neighbors one by one in a round robin fashion. 
The algorithm with a parameter $k$ is run on the directed spanner of the graph $G_k$ ($G$ without edges of latency $\geq$ $k$).  

\begin{algorithm}
\hspace{0.2cm} RR Broadcast ($k$)
\begin{algorithmic}[1]

\For{each vertex $v$ in parallel}
\For{iteration $i$ equals 1 to $(k\Delta_{out} + k)$}
\State propagate rumor set $R_v$ along the out-edges of length $\le k$ one-by-one in a round robin fashion %
\State add all received rumors to $R_v$
\EndFor
\EndFor
\end{algorithmic}
\caption{RR Broadcast}\label{alg:RRB}
\end{algorithm}

\begin{figure}[htb!]
\begin{center}
\includegraphics[scale=0.7]{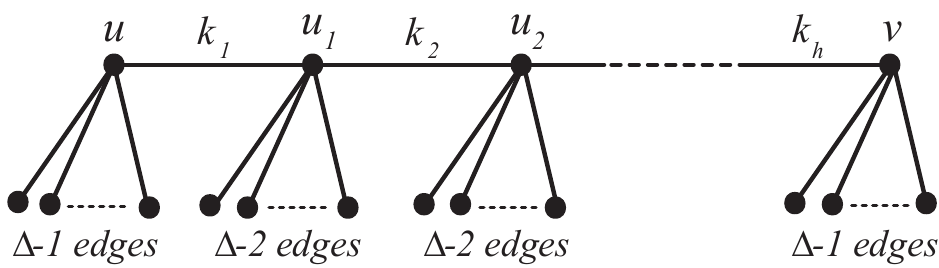}
\end{center}
\caption{Example of message propagation from node $u$ to $v$.\label{rel}}
\end{figure}

\begin{lemma}
	\label{lem:sb}
	After the execution of RR Broadcast algorithm with a parameter $k$ on the directed spanner of graph $G_k$, any two nodes $u$ and $v$ at a distance $\le k$ in $G$ have exchanged rumors with one another in $O(k\Delta_{out} + k)$ rounds, where $\Delta_{out}$ is the maximum out-degree of any node in $G_k$.
\end{lemma}

\begin{proof} Consider a path from a node $u$ to another node $v$ at a distance $k$ or less from it. Clearly, all edges in this path would have a weight of $\le k$. Therefore, we can work on $G_k$ ($G$ without edges of latency $\geq k$) as well without affecting the correctness of the algorithm. Also, let us assume that the number of hops between $u$ and $v$ to be $h$ which again would be $\le k$, since there are no fractional weights.  Let the latency between each hop be denoted by $k_i$ as shown in Figure \ref{rel}. Messages reach the next node when either of the nodes initiate a bidirectional exchange. For example,  $u$'s rumor could reach node $u_1$ either by a request initiated by node $u$ or by  $u_1$, depending upon the direction of the edge $uu_1$. In the worst case nodes have to try all other $\Delta_{out} -1$ links before initiating a connection along the required edge where $\Delta_{out}$ is the maximum out-degree of any node. After a connection is initialized it takes $k_1$ time to exchange rumors. By generalization, we observe that in the non-blocking model, the delay that can be incurred before rumor exchange among any two adjacent nodes $u_i$ and $u_{i-1}$  can be $\Delta_{out} + k_i$ in the worst case. In this way $u$'s rumor proceeds towards $v$ in individual steps, each step incurring a maximum cost of $\Delta_{out} + k_i$. A node might receive multiple rumors to propagate in the next round, which its adds to its rumor set and forwards to its neighbors in a round robin fashion. %
As such, the total worst case delay in rumor exchange among node $u$ and $v$ would be represented by 
\[
  \sum_{i=1}^{h} (\Delta_{out} + k_i) = h \Delta_{out} +  \sum_{i=1}^{h} k_i. 
\]
  But we know that both $h$ and $\sum_{i=1}^{h} k_i$ can have a maximum value equal to $k$ . Therefore, we conclude that for any two nodes $v$ and $u$ in $G_k$, $v$'s rumor would have reached $u$ and $u$'s rumor would have reached $v$ if all nodes forward rumors in a round robin fashion for $(k\Delta_{out} + k)$ rounds.
\end{proof}

Here, on the created spanner with stretch of $O(\log n)$, the maximum distance between any two nodes can be $O(D\log n)$. %
Since the maximum out-degree ($\Delta_{out}$) is $O(\log n)$ w.h.p., we get the following corollary.

\begin{corollary}
	The RR broadcast algorithm on the constructed spanner takes $O(D \log^2 n)$ time and solves all-to-all information dissemination w.h.p.
\end{corollary}

\noindent We combine all the previously defined techniques to a single algorithm called Spanner Broadcast.%

\begin{algorithm}
\hspace{0.2cm} Spanner Broadcast ($D$)
\begin{algorithmic}[1]
\For{each vertex $v$ in parallel}
\For{iteration $i$ = 1 to $O(\log n)$}   
\State Perform $D$-DTG   
\Statex  /* to gain neighborhood information */
\EndFor
\State call Spanner Construction procedure   
\Statex  /* executed locally */
\State call algorithm RR Broadcast ($O(D\log n)$)
\EndFor
\end{algorithmic}
\caption{Spanner Broadcast: for known diameter $D$}\label{alg:DSpanner Broadcast}
\end{algorithm}

\begin{lemma}
	\label{lem_Spanner Broadcast}
	For a graph $G$ with diameter $D$, Spanner Broadcast algorithm takes $O(D \log^3n)$ time for solving all-to-all information dissemination w.h.p. when $D$ is known to all the nodes.  %
\end{lemma}

\subsubsection{Unknown Diameter}
\label{sec:diameter}
For unknown diameter, we apply the standard guess-and-double strategy: begin with an initial guess of $1$ for $D$. Try the algorithm and see if it succeeds.  If so, we terminate.  Otherwise,  double the estimate and repeat. 
The challenge here is to correctly determine the termination condition i.e. how does a particular node determine whether information dissemination has been achieved for all other nodes. Early termination might lead to partial dissemination whereas late termination might cause the time complexity to increase. 

The critical observation is as follows: if two nodes $u$ and $v$  cannot communicate in one execution of all-to-all information dissemination (protocol RR Broadcast) for a given estimate of the diameter, then there must be some edge $(w,z)$ on the path from $u$ to $v$ where, in one execution: $u$ is able to communicate with $w$ but not with $z$.  There are two cases: If $w$ is not able to communicate with $z$, then it is aware that it has an unreachable neighbor and can flag the issue; the next time that $u$ and $w$ communicate, node $u$ learns of the problem.  Otherwise, if $w$ can communicate with $z$, then the next time that $u$ and $w$ communicate, node $u$ learns that there was a node it did not hear from previously.  In either case, $u$ knows that the estimate of $D$ was not correct and should continue.
Each node also checks whether it has heard from all of its neighbors, and raises an error flag if not.  We then repeat all-to-all broadcast so that nodes can check if everyone has the same ``rumor set" and that no one has raised an error flag.  In total, checking termination has asymptotic complexity of $O(D \log^2 n)$.

The Termination$\_$Check algorithm checks for every node that $v$ contacts or is contacted by (either directly or indirectly) whether that node has (i) exactly the same rumor set as $v$ and (ii) the value $0$ as its flag bit. The flag bit of a node is set to 1 if a neighbor of that node is not present in its rumor set or if the node has not yet exchanged all the rumors known to it presently with all of its neighbors in $G$ that are at a distance $\le$ to the current estimate of $D$ (say $k$): this condition is easily checked by either doing an additional $k$-DTG (which does not affect the complexity) or can be checked in parallel with the execution of RR Broadcast. If both of the above conditions are not met, then node $v$ sets its status to \emph{``failed''} and $v$ uses a broadcast algorithm for propagating the \emph{``failed''} message. Any broadcast algorithm that, given a parameter $k$, is able to broadcast and collect back information from all nodes at a distance $\le k$ from $v$, can be used. It is easily seen that RR Broadcast satisfies this criteria and can be used in this case. Note that broadcast is achieved here (for Spanner Broadcast algorithm) by execution of RR Broadcast, however when Pattern Broadcast algorithm (described later) invokes Termination$\_$Check, broadcast is achieved by execution of the sequence $T(k)$ (also described later). Here, the rumor set known to a particular vertex $v$ is denoted by $R_v$, $\Gamma(v)$ represents all its neighbors in $G$ whereas $k$-neighbors refers to only those nodes that are connected with $v$ with an edge of latency $k$ or less. Also, initially node$\_$status of all nodes is set to \emph{``default''}.

\begin{algorithm}
\hspace{0.2cm} Termination$\_$Check ($k$)
\begin{algorithmic}[1]
\If {(node $w \in \Gamma(v)$ and $w \notin R_v$) or (node $v$ has not exchanged rumors with all $k$-neighbors)}
\State set flag bit, $v_{flag} = 1$ 
\Else { set flag bit, $v_{flag} = 0$ }
\EndIf 
\State broadcast and gather all responses from any node $u$ in $v$'s $k$-distance neighborhood %
\If{$\exists$ any $u$ such that ($R_v \neq R_{u}$) or ($u_{flag}=1$)}
\State set node$\_$status = \emph{``failed''}
\EndIf
\State broadcast  \emph{``failed''} message to the $k$-distance neighborhood %
\If {received message = \emph{``failed''}}
\State set node$\_$status = \emph{``failed''}
\EndIf 
\end{algorithmic}
\caption{Termination$\_$Check}\label{alg:TC}
\end{algorithm}

\noindent We prove the following regarding the termination detection:
\begin{lemma}
	\label{termination}
	No node terminates until it has exchanged rumors with all other nodes.  Moreover, all nodes terminate in the exact same round.
\end{lemma}

\begin{proof} Suppose that a node $v$ terminates without having exchanged rumors with some other node $w$. Considering any path from node $v$ to node $w$, let $u$ be the farthest node (in hop distance) with which $v$ has exchanged rumors with and let $x$ be the next node in the path.\\
\emph{Case 1 }: $u$ has exchanged rumors with $x$.
It implies that $v$ has also exchanged rumors with $x$, from the condition that all nodes that exchange rumors with one another have the same rumor set. Thus, contradicting the fact that $u$ is the farthest node on the path that $v$ has exchanged rumors with.
\\
\emph{Case 2 }: $u$ has not exchanged rumors with $x$.
If $u$ had not exchanged rumors with $x$, then $u$ would have set its flag bit as 1, which would have been detected by $v$ during the broadcast and it would not have terminated. This also gives us a contradiction.
Thus, no such node $w$ exists and $v$ terminates only after it has exchanged rumors with all the other nodes.

For the second part of the proof, let consider $u$ and $v$ to be nodes such that
$v$ is set for termination and has not set its status to \emph{``failed''} in the Termination$\_$Check algorithm, whereas, in the same iteration, node $u$ has set its status to \emph{``failed''} and hence is set to continue. We show that there cannot be two such nodes in the same round. The node $v$ did not set its status to \emph{``failed''} implying all the nodes that it exchanged rumors with had exactly the same set of rumors, none of the nodes had set its flag bit as 1 and in addition it did not receive a \emph{``failed''} message from any other node. From the first part, we know that the set of nodes that $v$ exchanged rumors with is the entire vertex set of the graph $G$. That implies, $v$ has also exchanged rumors with $u$: node $u$ also has the exact set of rumors (which essentially is all the rumors from all the nodes) and does not have a set flag bit. So in the current iteration, if any other node broadcasted a \emph{``failed''} message both $v$ and $u$ would have received it resulting in both nodes to set their status as \emph{``failed''}. Again, since the rumor sets of both nodes are identical, both nodes would observe the same flag bits of all the nodes. Then node $u$ will also not satisfy the termination condition and will not set its status as \emph{``failed''}. This gives us a contradiction that completes the proof.
\end{proof}

\begin{algorithm}
\hspace{0.2cm} Spanner Broadcast ($k$)
\begin{algorithmic}[1]
\State $k$=1

\MRepeat 
\State call algorithm \emph{Spanner Broadcast ($k$)}

\State call algorithm \emph{Termination$\_$Check ($k$)}  
\If{node$\_$status = \emph{``failed''}}
\State $k$ = $2 k$
\State set node$\_$status to \emph{``default''}
\Else { terminate}
\EndIf
\EndMRepeat
\end{algorithmic}
\caption{Spanner Broadcast: for unknown diameter.}\label{alg:Spanner Broadcast}
\end{algorithm}

Combining the all-to-all dissemination protocol with the termination detection, we get the following:
\begin{theorem}
	\label{thm:main2}
	There exists a randomized gossip algorithm that solves the all-to-all information dissemination problem w.h.p. and terminates in $O(D \log^3 n)$ rounds.
\end{theorem}

\subsection{Pattern Broadcast Algorithm}
\label{app:altalgo}
We propose an alternate deterministic pattern based broadcast algorithm to solve all-to-all information dissemination without any global knowledge (i.e., knowledge of a polynomial upper bound on $n$ is not required) that takes $O(D\log^2n\log D)$ time. This algorithm %
works even when nodes cannot initiate a new exchange in every round, and wait till the acknowledgement of the previous message, i.e., communication is blocking.

The algorithm involves repeatedly invoking the $\ell$-DTG algorithm with different parameters determined by a particular pattern. The intuition behind the choice of the pattern is to make minimal use of the heavier latency edges by collecting as much information as possible near the heavier latencies before making use of that edge. The pattern for $k$ is derived according to a sequence $T(k)$ that is recursively defined as follows:  
\begin{align*}
T(1) &= 1\text{-DTG}\\ 
T(2) &= T(1) \cdot 2\text{-DTG} \cdot T(1)\\
T(4) &= T(2) \cdot 4\text{-DTG} \cdot T(2)\\
&\vdots\\ 
T(k) &= T(k/2) \cdot k\text{-DTG} \cdot T(k/2)
\end{align*}
We show that, when the above sequence is run for the particular pattern for length $k$, it guarantees that any node $u$ and $v$ in the graph $G$, at a distance of $\le k$, have exchanged their rumors with one another. 
Overall, the pattern of values of the parameter $\ell$ is 
\[
1,2,1,4,1,2,1,8,1,2,1,4,1,2,1,\dots,k,\dots,1,2,1,4,1,2,1,8,1,2,1,4,1,2,1,
\]

and, for each value $\ell$, we perform the $\ell$-DTG protocol. 
That is, $T(k)$ is a sequence of calls to $\ell$-DTG with varying parameters according to a known pattern. %

\begin{lemma}
After the execution of $T(k)$, any node in the weighted graph G (V,E) has exchanged rumors with all other nodes that are at distance $k$ or less from it.
\end{lemma}
\begin{proof}
We proceed by induction over the path length $k$.
For the base case, recall from \cite{Haeupler:2013:SFD:2627817.2627868}
that, after running $T(1)$  on $G_1$ (subgraph of $G$ induced by edges with latency $\le 1$),  any node $v$ has exchanged rumors with all its distance 1 neighbors. 

For the inductive step, suppose that the claim is true for $T(k)$, i.e. after running the sequence, any node $v$ has exchanged rumors with all other nodes at a weighted distance $\le k$.
To prove the claim for $T(2k)$ (i.e. $T(k) \cdot  2k\text{-DTG} \cdot T(k)$), we consider various possibilities of forming a path of length $2k$.

\case{Case 1} The path consists only of edges with latencies $\le k$. Here we distinguish two sub-cases: \\
\case{Case 1a} There exists a node $m$ which is equidistant from both end points $u$ and $v$ (see Figure \ref{fig_sim1}).  %
By the induction hypothesis, both nodes $u$ and $v$ would have exchanged rumors with node $m$ in the initial $T(k)$. In the next $T(k)$, node $m$ propagates all rumors that it received from $u$ to $v$ and vice-versa. 
\begin{figure}[h]
\centering
\includegraphics[width=3in]{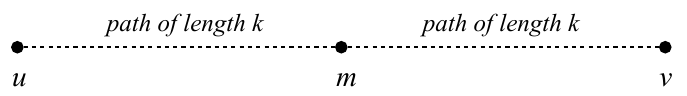}
\caption{Case 1a}
\label{fig_sim1}
\end{figure}

\case{Case 1b} No such node middle exists as depicted in Figure \ref{fig_sim2}. Then, after the initial $T(k)$, node $u$ must have exchanged rumors with $m_1$ and node $v$ with $m_2$, due to the induction hypothesis. In the invocation of the $2k$-DTG, node $m_1$  propagates all rumors gained from $u$ to $m_2$, and $m_2$ also propagates all rumors gained from $v$ to $m_1$. This information then travels from $m_1$ to $u$ and from $m_2$ to $v$ in the final $T(k)$.
\onlyLong{
\begin{figure}[h]
\centering
\includegraphics[width=3in]{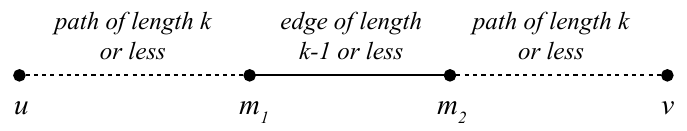}
\caption{Case 1b}
\label{fig_sim2}
\end{figure}
}

\case{Case 2} 
There exists at most one edge $e$ with latency value in between $[k+1,2k]$.
This situation can yield one of the following two sub-cases: \\
\case{Case 2a} Edge $e$ is located at one end of the path (see Figure \ref{fig_sim3}). By  the induction hypothesis, node $v$ would have exchanged rumors with $m$ in the initial $T(k)$. In the $2k$-DTG, $u$ gets to know this (and other) rumors from $m$ and $m$ also gets to know $u$'s rumors. In the next $T(k)$, node $m$ propagates all rumors gained from $u$ to $v$. 

\onlyLong{
\begin{figure}[h]
\centering
\includegraphics[width=3in]{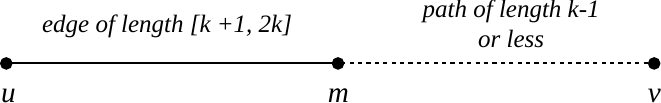}
\caption{Case 2a}
\label{fig_sim3}
\end{figure}
}

\case{Case 2b} The edge is located between two inner nodes on the path (see Figure \ref{fig_sim4}). In this case, by  the induction hypothesis, node $u$ has exchanged rumors with $m_1$, whereas node $v$ has exchanged rumors with node $m_2$ in the initial $T(k)$. In the $2k$-DTG, node $m_1$  propagates all rumors gained from $u$ to $m_2$. Moreover, $m_2$  propagates all rumors gained from $v$ to $m_1$. These rumors then propagate from $m_1$ to $u$ and from $m_2$ to $v$ in the final $T(k)$.
\onlyLong{
\begin{figure}[h]
\centering
\includegraphics[width=3in]{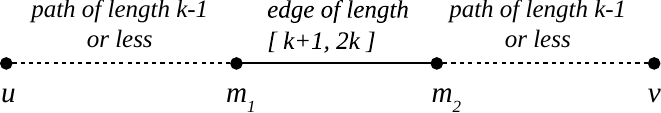}
\caption{Case 2b}
\label{fig_sim4}
\end{figure}
}
\end{proof}

\begin{lemma}
For known diameter, solving all-to-all information dissemination by executing the sequence $T(D)$, takes $O(D \log^2 n \log D)$ time.
\end{lemma}
\begin{proof}
From the way the sequence is constructed, we observe the recurrence relation 
$T(k) = 2T(k/2) + k\log^2 n$.
Using standard methods to solve the recurrence completes the proof.
\end{proof}

When the graph diameter is known to all nodes, nodes can just invoke $T(D)$ to solve all-to-all information dissemination.
For completeness, we also present an algorithm called Pattern Broadcast that uses the sequence of invocations of $\ell$-DTG to solve all-to-all information dissemination, when the graph diameter is unknown. This algorithm is similar in flavour to that of the Spanner Broadcast algorithm described in Section \ref{sec:diameter} and also makes use of the Termination$\_$Check algorithm, albeit with a different broadcasting technique (calling $T(k)$ rather than RR Broadcast).

\begin{algorithm}
\hspace{0.2cm} Pattern Broadcast ($k$)
\begin{algorithmic}[1]
\State $k$=1
\MRepeat 
\State execute sequence $T(k)$
\State call algorithm \emph{Termination$\_$Check ($k$)}  
\If{node$\_$status = \emph{``failed''}}
\State $k$ = $2 k$
\State set node$\_$status to \emph{``default''}
\Else { terminate}
\EndIf
\EndMRepeat
\end{algorithmic}
\caption{Pattern Broadcast; code for vertex $v$.}\label{alg:PD}
\end{algorithm}

\begin{lemma}
There exists a deterministic gossip algorithm that solves the all-to-all information dissemination problem and terminates in $O(D \log^2n \log D)$ rounds.
\end{lemma}
\noindent Applying techniques similar to Section \ref{sec:diameter}, similar results can be shown for the case with unknown diameter as well.

\section{Algorithms for Unknown Latencies} \label{sec:upper_bounds}
We divide the upper bounds on information dissemination into two sub-components and later combine them to obtain a unified result. First, we analyze 
classical push-pull, showing that it completes in time $O( \frac{\ell_*\log n}{\phi_*})$, which is optimal when $D+\Delta$ is large. Alternatively for graphs where $D+\Delta$ is small, we give an algorithm wherein each node first spends $\tilde O(D+\Delta)$ time discovering the neighboring latencies after which nodes use the local information to build a spanner, across which data can be distributed in $\tilde O(D)$ time.

\subsection{Push-Pull}

To show the time required for information dissemination in a weighted graph $G$ using push-pull, we define  $E_\ell$ as the \emph{set of all edges of latency $\le \ell$}, $E_u$ as the \emph{set of incident edges of vertex $u$} and $E_{u,\ell} = E_\ell \cap E_u$. 

\begin{theorem} \label{thm:PUSHPULLUpper}
	The push-pull protocol achieves information dissemination w.h.p. in $O( \frac{\ell_*\log n}{\phi_*})$ rounds in a network $G$, where $\phi_*$ is the critical weighted conductance of $G$ and $\ell_*$ is the corresponding critical latency. 
\end{theorem}

\begin{proof}
From the given weighted graph $G$, %
we construct a {strongly edge-induced graph} $G_{\ell}$, which is a generalization of the strongly (vertex) induced subgraph defined in \cite{Censor-Hillel:2012:GCP:2213977.2214064} and which has the same vertex set as $G$.  The edges of $G_{\ell}$ have a multiplicity\footnote{The ``multiplicity of an edge'' is called ``edge weight'' in \cite{Censor-Hillel:2012:GCP:2213977.2214064}. We use a different terminology here to avoid confusion with the latencies of edges and consider ``edge weight'' as a synonym to edge latency instead.} defined by the \emph{edge multiplicity function} $\mu$, given by 
\begin{equation}
    \mu(u,v)=
    \begin{cases}
       	1 & \text{if}\ (u,v) \in E_{\ell}; \\
      	|E_u| - |E_{u,\ell}| & \text{if}\ u=v; \\
      	0 & \text{otherwise.}\\
    \end{cases} 
  \end{equation}

The \emph{informed node set} refers to the set of vertices that are in possession of some message $m$ originating from a vertex $s$ when running push-pull. %
When executing the push-pull process on $G_\ell$, each message takes $1$ round to traverse an edge, and hence a message sent in $G_{\ell}$ can be simulated by at most $\ell$ rounds in $G$.
Let random variable $I_G(\ell \cdot r)$ refer to the informed node set in graph $G$ after running $\ell \cdot r$ rounds of push-pull on $G$; we can think of parameter $r$ as the number of push-pull rounds that we want to simulate on $G_{\ell}$.  Similarly, we define random variable $I_{G_\ell}(r)$ to be the informed node set in $G_\ell$ after $r$ rounds. 

Observe that each node $v$ selects an incident edge in $E_{\ell}$ from $G_{\ell}$ in the push-pull protocol with the same probability as in %
$G$. The probability of choosing an edge $\in E_u \setminus E_{\ell}$ (i.e., a self loop in case of $G_{\ell}$) is $\mu(u,u)/\sum_{v \in V}\mu(u,v)$ in both graphs. 
Clearly, choosing a self loop of a node $u$ cannot help in the propagation of the message in $G_{\ell}$, but choosing the corresponding edge in $G$ might. 

Now, consider the Markov chain process describing the informed node set, when running push-pull.
Formally, the state space of the Markov chain consists of all possible informed node sets. Only paths that correspond to monotonically growing informed node sets have nonzero probability. 

We will show by induction (over $r$) that the Markov process that describes the (monotonically growing) set of informed nodes on %
$G$ stochastically dominates the respective Markov process for the informed nodes in graph $G_{\ell}$.
Since we are simulating push-pull on graphs $G$ and $G_\ell$ having the same node set, we assume that exactly $1$ node has the initial message $m$, which means that both Markov chains start at the state representing the same (singleton) set of informed nodes.
Thus, for the induction base case ($r=0$), we have $\Prob{I_G(0\cdot \ell) = S_0} = \Prob{I_{G_\ell}(0)= S_0},$ for any node set $S_0$.

We now focus on the induction step. 
By the induction hypothesis, it holds that for any set $S_r$,
\begin{equation} 
	\Prob{ I_G(r\cdot \ell) = S_r} \geq \Prob{I_{G_\ell}(r) = S_r}.  \label{eq:ind}
\end{equation}
Consider any set $S_{r+1} \supseteq S_r$.
Let $\{S_r\ra_G S_{r+1}\}$ be the event that the Markov chain transits from $S_r$ to $S_{r+1}$ for graph $G$, and define $\{S_r\ra_{G_\ell} S_{r+1}\}$ similarly.
It follows that
\begin{flalign*}
&\Prob{ I_G((r+1)\cdot \ell) = S_{r+1} } \\
&=\sum_{\text{$S_r$}} \Prob{I_G(r\cdot \ell) = S_r} \cdot \Prob{S_r \ra_G S_{r+1}} \\
&\geq \sum_{\text{$S_r$}}\Prob{ I_{G_\ell}(r) = S_r} \cdot \Prob{S_r \ra_G S_{r+1}} \tag{by \eqref{eq:ind}} \\
&\geq \sum_{\text{$S_r$}}\Prob{I_{G_\ell}(r) = S_r} \cdot \Prob{S_r \ra_{G_\ell} S_{r+1}}\\ 
& \hspace{1.5cm}\hfill \text{(we obtained $G_\ell$ by making some edges self-loops)}\\
&= \Prob{ I_{G_\ell}(r+1) = S_{r+1}}
\end{flalign*}

Thus, it follows that %
the probability of reaching any informed node set $S$ by using the Markov chain in $G$ is at least as large as the probability of reaching the same set $S$ by using the Markov chain for $G_{\ell}$.  
To translate this result back to our actual network $G$ (with weighted edges), we charge each round of push-pull in $G_{\ell}$ to $\ell$ rounds in $G$.
It is easy to see that the (unweighted) conductance $\phi({G_{\ell}})$ corresponds to $\phi_{\ell}(G)$, as a self-loop at node $u$ is counted as $\mu(u,u)$ edges when computing the volume. %
From \cite{stacs2011_conductance} and \cite{Censor-Hillel:2012:GCP:2213977.2214064}  it is known that $O(\log(n) / \phi({G_{\ell}}))$ rounds suffice w.h.p. to solve broadcast in $G_{\ell}$. Hence, achieving broadcast in $G$ requires $O(\ell\log(n) / \phi_{\ell}({G}))$ rounds. 
Since the above analysis applies for any $\ell\ge 1$, and in particular for the critical latency $\ell_*$, the theorem follows. (When $\ell=\ell_*$, $\phi(G_{\ell_*}) = \phi_{\ell*}(G)= \phi_*$). 
\end{proof}

We combine Theorem \ref{thm:PUSHPULLUpper} with Theorem \ref{thm:conductance} to obtain the following corollary that gives the upper bound on information dissemination using push-pull in terms of $\phi_{avg}$. %

\begin{corollary} \label{coro:pushpull}
The push-pull protocol achieves broadcast w.h.p. in $O(\frac{\mathcal{L} \log n} {\phi_{avg}})$ rounds in a network $G$, where $\phi_{avg}$ is the average weighted conductance of $G$ and $\mathcal{L}$ is the number of non-empty latency classes in $G$. 

\end{corollary}

\subsection{Tweaked Spanner Broadcast Algorithm}

In Section \ref{spanner_algo} we provide an algorithm that solves all-to-all information dissemination when each node knows the latencies of all its adjacent edges and message size is unbounded. The same algorithm can be naturally extended for the case where nodes do not know the adjacent latencies by first discovering the edge latencies and then running the algorithm as such.

In the case where both $D$ and $\Delta$ are known, each node broadcasts a request to each neighbor (sequentially) for $\Delta$ rounds and then waits up to $D$ rounds for a response to determine the adjacent edge's latency.
In Section \ref{sec:diameter}, we show the guess and double strategy for the case where just the diameter $D$ is unknown. As we can efficiently detect when information dissemination has completed correctly, we can use a similar strategy to estimate $\Delta$ if only $\Delta$ is unknown or alternatively guess the value of $D+\Delta$ if both $D$ and $\Delta$ are unknown. %
By arguments similar to Section~\ref{sec:diameter}, we show that the guessing and doubling strategy does not increase the overall time complexity. Therefore, we obtain an algorithm that solves information dissemination in $O((D+\Delta)\log^3n)$ time.

Additionally, if edge latencies are unknown, we can obtain similar results for the Pattern Broadcast algorithm (see Alg.~\ref{alg:PD}) as well by using the guess and double strategy.

\section{Unified Upper Bounds}
\label{sec:unified}

Combining the results shown above, we can run both push-pull and the spanner algorithm in parallel to obtain unified upper bounds for both the known and the unknown latencies cases. However, we point out that, for information dissemination, push-pull works with small message sizes whereas the spanner algorithm does not (because of its reliance on DTG). Also, exchanging messages with the help of the spanner does not have good robustness properties whereas push-pull is inherently quite robust.  For graphs that have small diameters, we can use the alternative pattern based algorithm as compared to the spanner based one. However, here we give our unified upper bounds based on the spanner algorithm of Section \ref{sec:spanner}.
 
\begin{theorem}
	\label{thm:upper}
	There exists a randomized gossip algorithm that solves the all-to-all information dissemination problem in $O(\min((D + \Delta)\log^3{n}, (\ell_*/\phi_*)\log n)$  time when latencies are not known and in $O(\min(D\log^3{n}, (\ell_*/\phi_*)\log n ))$ time when latencies are known.
\end{theorem}

\begin{corollary}
	\label{thm:upperavg}
	There exists a randomized gossip algorithm that solves the all-to-all information dissemination problem in $O(\min((D + \Delta)\log^3{n}, (\mathcal{L}/\phi_{avg})\log n)$  time when latencies are unknown and in $O(\min(D\log^3{n}, (\mathcal{L}/\phi_{avg})\log n ))$ time when latencies are known.%
\end{corollary}

\section{Conclusion}

We have presented two different concepts, namely the critical and the average weighted conductance, that characterize the bottlenecks in communication for weighted graphs.  We believe that these parameters will be useful for a variety of applications that depend on connectivity.

A question that remains is whether the running time of $O(D \log^3n)$ for information dissemination can be improved, e.g., using better spanner constructions or more efficient local broadcast to save the polylogarithmic factors.  (Recall that in the unweighted case, there are information dissemination protocols that run in $O(D + \polylog{n})$ time.) Another interesting direction would be the development of reliable robust fault-tolerant algorithms in this regard.

Another issue is whether we can reduce the number of \emph{incoming} messages in a round; recently, Daum et al. \cite {DBLP:journals/corr/DaumKM15} have considered such a more restricted model, yielding interesting results. It would also be interesting to look at the bounds where each node is only allowed $O(1)$ connections per round, whether initiated by the node itself or by its neighbor.

\section*{Acknowledgments}
This research was supported by AcRF Tier $1$ grant T1 251RES1719 (Adaptive Data Structures: Concurrent, Cache-Efficient, Distributed) and the Natural Sciences and Engineering Research Council of Canada (NSERC).
The authors would like to thank George Giakkoupis for the helpful conversations and useful ideas.

\appendix
\section{Appendix}
\subsection{The DTG Local Broadcast Protocol}
\label{app:dtg}

In this section, we describe in more detail the DTG protocol that was originally developed in~\cite{Haeupler:2013:SFD:2627817.2627868} as well as the $\ell$-DTG algorithm.

It is clear that the algorithm solves local broadcast because it keeps on contacting new neighbors until it has exchanged rumors with all of its neighbors. The author \cite{Haeupler:2013:SFD:2627817.2627868} makes use of binomial trees to derive the time complexity and better explain the working of the algorithm.

The key idea used for deriving the time complexity is to show that when information is propagated in a pipelined manner along the binomial trees (created on-the-fly), then for any node that is still active in the $i^{th}$ iteration, it has a binomial tree of order $2^{i}$ ($i$-tree of depth $i$: see Figure \ref{fig_itree}) rooted at it. Furthermore, it is shown that for any two different nodes that are still active in iteration $i$, their $i$-trees are vertex disjoint. Since an $i$-tree is formed by joining two $(i-1)$-trees, the growth rate of an $i$-tree is exponential which limits the number of iterations to $O(\log n)$. Also, each node on an average needs to contact $O(\log n)$ nodes ($O(i)$ nodes in the $i^{th}$ round). Thus, the overall complexity of the algorithm becomes $O(\log^2 {n})$. In our case, for $\ell$-DTG, the additional waiting time of $\ell$ increases the time complexity to $O(\ell \log^2 {n})$.

\begin{figure}[h]
\centering
\includegraphics[scale=0.6]{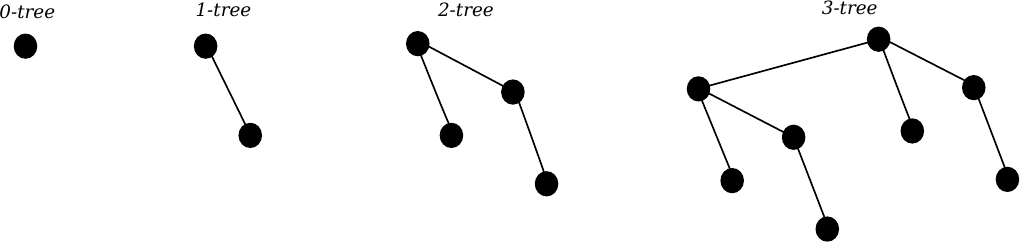}
\caption{$i$-trees for $i \in {0,1,2,3}$}
\label{fig_itree}
\end{figure}

The $i$-tree can be seen as witness structures that provides an explanation as to why a node was active in that particular iteration. 
The $i$-tree rooted at a particular node is built recursively as the rounds progress and essentially store the information about which other nodes communicated with one another in which particular round as viewed from the root node. For example, in Figure \ref{fig_5tree} , the labels on the edges denote the time in which the node of the higher level contacted the lower level node (as observed by the root node). The root contacts the nodes in first level in rounds according to their label, the nodes on the first level similarly contact the nodes in the second level in rounds according to their label and so on. This observation also helps in the realization of the key idea of a node being active in the $i^{th}$ round having an $i$-tree rooted at it. The nodes in the first level did not contact the root previously as they were busy contacting the nodes of the second level, the nodes of the second level did not contact nodes on the first level as they were busy contacting the nodes in the third level and so on.    %

\begin{figure}[h]
\centering
\includegraphics[scale=0.6]{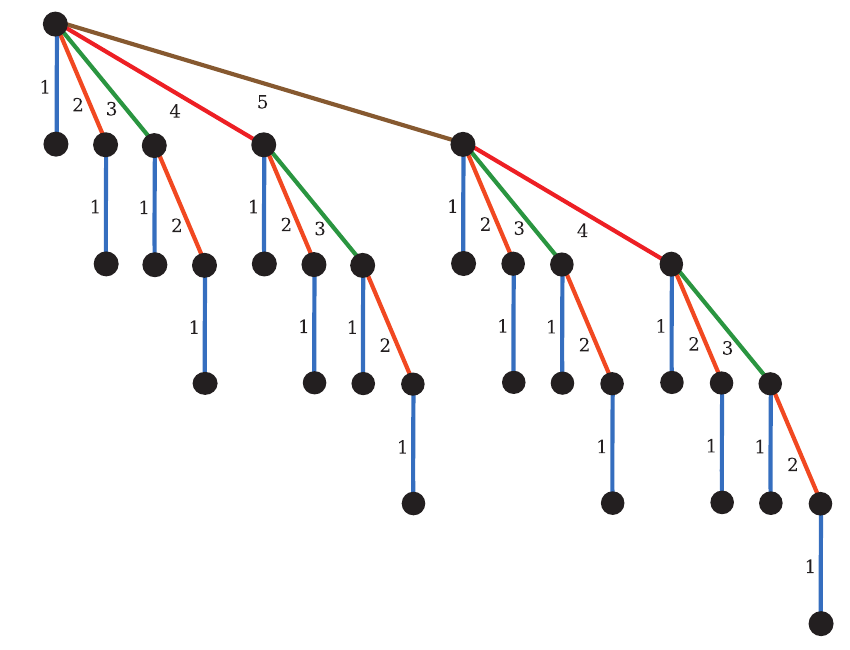}
\caption{$5$-tree with edge labels}
\label{fig_5tree}
\end{figure}

As shown in the pseudo code, in the initial PUSH sequence, the message is propagated in a decreasing order of connection round number (as observed by the root node: given by the labels on the edges of Figure \ref{fig_itree}), helping in pipe-lining the roots message to all other nodes of the $i$-tree. Similarly, during the initial PULL sequence the message from the nodes is pipelined up to the root. The subsequent PULL-PUSH sequence helps in maintaining the symmetry of the algorithm such that if node $u$ learns about node $v$, then node $v$ also learns about node $u$. Finally, the collection of rumors $R$ is updated to the union of rumors collected in the aforementioned sequences.

For $\ell$ being an integer $\ge$ 1, we run the modified DTG algorithm on a sub-graph of $G$, $G_{\ell}$, rather than on $G$, where $G_{\ell}$ contains only the edges of length up to $\ell$. Lets denote this algorithm as $\ell$-DTG. The algorithm is presented below and each node $v$ belonging to $G_{\ell}$ runs it in parallel. $\Gamma$($v$) can be considered  as the neighborhood of $v$ comprising of set of nodes that are node $v$'s $1$-hop neighbors.

\begin{algorithm}
\setlength{\columnsep}{10pt}
\hspace{0.2cm} $\ell$-DTG ($\ell$)
\begin{algorithmic}[1]
\State $R$ = $v$
\For{$i$ = 1 UNTIL $\Gamma$($v$)$\backslash R$ = $\phi$}
\State link to any new neighbor $u_i$ $\in$ $\Gamma$($v$)
\State $R'$ = $v$
\State PUSH : \For {$j$ = $i$ downto 1}
\State send rumors in $R'$ to $u_j$
\State wait for $\ell$ time to receive $u_j$'s rumors
\State add all received rumors to $R'$
\EndFor
\State PULL: \For {$j$ = 1 to $i$}
\State send rumors in $R'$ to $u_j$
\State wait for $\ell$ time to receive $u_j$'s rumors
\State add all received rumors to $R'$
\EndFor
\State $R''$ = $v$
\State perform PULL, PUSH with $R''$
\State $R$ = $R'$ $\cup$ $R''$ 
\EndFor

\end{algorithmic}
\caption{$\ell$-DTG}\label{alg:MDTG}
\end{algorithm}

\bibliographystyle{plain}
\bibliography{bibliography}

\begin{thebibliography}{10}

\bibitem{DBLP:conf/focs/AugustineP0RU15}
John Augustine, Gopal Pandurangan, Peter Robinson, Scott Roche, and Eli Upfal.
\newblock Enabling robust and efficient distributed computation in dynamic
  peer-to-peer networks.
\newblock In {\em {FOCS}}, 2015.

\bibitem{baswana}
Surender Baswana and Sandeep Sen.
\newblock A simple and linear time randomized algorithm for computing sparse
  spanners in weighted graphs.
\newblock {\em Random Structures and Algorithms}, 30(4):532--563, 2007.

\bibitem{Boyd:2006:RGA:1148663.1148679}
Stephen Boyd, Arpita Ghosh, Balaji Prabhakar, and Devavrat Shah.
\newblock Randomized gossip algorithms.
\newblock {\em IEEE/ACM Trans. Netw.}, 14(SI):2508--2530, June 2006.

\bibitem{Bradonjic:2010:EBR:1873601.1873715}
Milan Bradonji\'{c}, Robert Els\"{a}sser, Tobias Friedrich, Thomas Sauerwald,
  and Alexandre Stauffer.
\newblock Efficient broadcast on random geometric graphs.
\newblock In {\em SODA}, 2010.

\bibitem{Censor-Hillel:2012:GCP:2213977.2214064}
Keren Censor-Hillel, Bernhard Haeupler, Jonathan Kelner, and Petar Maymounkov.
\newblock Global computation in a poorly connected world: Fast rumor spreading
  with no dependence on conductance.
\newblock In {\em STOC}, 2012.

\bibitem{weak_conductance}
Keren Censor-Hillel and Hadas Shachnai.
\newblock Fast information spreading in graphs with large weak conductance.
\newblock In {\em SODA}, 2011.

\bibitem{Chierichetti:2010:ATB:1806689.1806745}
Flavio Chierichetti, Silvio Lattanzi, and Alessandro Panconesi.
\newblock Almost tight bounds for rumour spreading with conductance.
\newblock In {\em STOC}, 2010.

\bibitem{Chierichetti:2010:RSG:1873601.1873736}
Flavio Chierichetti, Silvio Lattanzi, and Alessandro Panconesi.
\newblock Rumour spreading and graph conductance.
\newblock In {\em SODA}, 2010.

\bibitem{DBLP:journals/corr/DaumKM15}
Sebastian Daum, Fabian Kuhn, and Yannic Maus.
\newblock Rumor spreading with bounded in-degree.
\newblock In {\em SIROCCO}, pages 323--339. Springer, 2016.

\bibitem{Demers:1987:EAR:41840.41841}
Alan Demers, Dan Greene, Carl Hauser, Wes Irish, John Larson, Scott Shenker,
  Howard Sturgis, Dan Swinehart, and Doug Terry.
\newblock Epidemic algorithms for replicated database maintenance.
\newblock In {\em PODC}, 1987.

\bibitem{Doerr:2011:SNS:1993636.1993640}
Benjamin Doerr, Mahmoud Fouz, and Tobias Friedrich.
\newblock Social networks spread rumors in sublogarithmic time.
\newblock In {\em STOC}, 2011.

\bibitem{Doerr:2012:WRS:2184319.2184338}
Benjamin Doerr, Mahmoud Fouz, and Tobias Friedrich.
\newblock Why rumors spread so quickly in social networks.
\newblock {\em Commun. ACM}, 55(6):70--75, June 2012.

\bibitem{FernandezAnta2012}
Antonio F.~Anta, Alessia Milani, Miguel~A. Mosteiro, and Shmuel Zaks.
\newblock Opportunistic information dissemination in mobile ad-hoc networks:
  the profit of global synchrony.
\newblock {\em Distributed Computing}, 25(4):279--296, Aug 2012.

\bibitem{FARACHCOLTON201360}
Mart{\'i}n Farach-Colton, Antonio~F. Anta, and Miguel~A. Mosteiro.
\newblock Optimal memory-aware sensor network gossiping (or how to break the
  broadcast lower bound).
\newblock {\em Theoretical Computer Science}, 472:60 -- 80, 2013.

\bibitem{firege}
Uriel Feige, David Peleg, Prabhakar Raghavan, and Eli Upfal.
\newblock Randomized broadcast in networks.
\newblock In {\em Algorithms}, volume 450 of {\em Lecture Notes in Computer
  Science}, pages 128--137. Springer, 1990.

\bibitem{Fraigniaud:2010:BCC:1810479.1810505}
Pierre Fraigniaud and George Giakkoupis.
\newblock On the bit communication complexity of randomized rumor spreading.
\newblock In {\em SPAA}, 2010.

\bibitem{4447471}
R.~Gandhi, A.~Mishra, and S.~Parthasarathy.
\newblock Minimizing broadcast latency and redundancy in ad hoc networks.
\newblock {\em Networking, IEEE/ACM Transactions on}, 16(4):840--851, Aug 2008.

\bibitem{stacs2011_conductance}
George Giakkoupis.
\newblock Tight bounds for rumor spreading in graphs of a given conductance.
\newblock In {\em 28th International Symposium on Theoretical Aspects of
  Computer Science (STACS)}, March~10--12 2011.

\bibitem{Giakkoupis2014}
George Giakkoupis, Thomas Sauerwald, and Alexandre Stauffer.
\newblock Randomized rumor spreading in dynamic graphs.
\newblock In {\em Automata, Languages, and Programming: 41st International
  Colloquium, ICALP}, 2014.

\bibitem{Haeupler:2013:SFD:2627817.2627868}
Bernhard Haeupler.
\newblock Simple, fast and deterministic gossip and rumor spreading.
\newblock In {\em SODA}, 2013.

\bibitem{Haeupler:2014:OGD:2611462.2611489}
Bernhard Haeupler and Dahlia Malkhi.
\newblock Optimal gossip with direct addressing.
\newblock In {\em PODC}, 2014.

\bibitem{hoory06}
Shlomo Hoory, Nathan Linial, and Avi Wigderson.
\newblock Expander graphs and their applications.
\newblock {\em Bull. Amer. Math. Soc.}, 43(04):439--562, 2006.

\bibitem{Jerrum:1988:CRM:62212.62234}
Mark Jerrum and Alistair Sinclair.
\newblock Conductance and the rapid mixing property for markov chains: The
  approximation of permanent resolved.
\newblock In {\em STOC}, 1988.

\bibitem{Karp}
R.~Karp, C.~Schindelhauer, S.~Shenker, and B.~Vocking.
\newblock Randomized rumor spreading.
\newblock In {\em FOCS}, 2000.

\bibitem{Kempe:2001:SGR:380752.380796}
David Kempe, Jon Kleinberg, and Alan Demers.
\newblock Spatial gossip and resource location protocols.
\newblock In {\em STOC}, 2001.

\bibitem{Mosk-Aoyama:2006:CSF:1146381.1146401}
Damon Mosk-Aoyama and Devavrat Shah.
\newblock Computing separable functions via gossip.
\newblock In {\em PODC}, 2006.

\bibitem{DBLP:conf/wdag/Newport14}
Calvin Newport.
\newblock Radio network lower bounds made easy.
\newblock In {\em {DISC}}, 2014.

\bibitem{5062132}
A.D. Sarwate and A.G. Dimakis.
\newblock The impact of mobility on gossip algorithms.
\newblock In {\em INFOCOM}, April 2009.

\bibitem{vadhan}
Salil~P Vadhan.
\newblock {Pseudorandomness}.
\newblock {\em Foundations \& Trends in Theoretical Computer Science},
  7(1-3):1--336, 2012.

\end{thebibliography}

\end{document}